\documentclass[11pt]{article}

\usepackage{amssymb}
\usepackage{amstext}
\usepackage{amsfonts}
\usepackage{amsmath}
\usepackage{amsthm}
\usepackage{graphicx}
\usepackage{mathtools}

\newtheorem{theorem}{Theorem}[section]    
\newtheorem{definition}{Definition}[section] 
\newtheorem{corollary}[theorem]{Corollary}    
\newtheorem{proposition}[theorem]{Proposition}    
\newtheorem{lemma}[theorem]{Lemma}    
\renewcommand{\qed}{\hfill{$\rule{6pt}{6pt}$}} 
\renewenvironment{proof}{\noindent{\bf Proof:}}{\qed\\}
\newenvironment{proofof}[1]{\noindent{\bf Proof of #1:}}{\qed\\}

\numberwithin{equation}{section} 

\newcommand{\complex}{{\mathbb C}}

\newcommand{\ket}[1]{| #1 \rangle}
\newcommand{\bra}[1]{ \langle #1 |}
\newcommand{\ketbra}[2]{| #1 \rangle\!\langle #2 |}

\newcommand{\norm}[1]{\left\| #1 \right\|}
\newcommand{\trnorm}[1]{\left\| #1 \right\|_{\mathrm{tr}}}
\newcommand{\size}[1]{\left| #1 \right|}
\newcommand{\set}[1]{\left\{ #1 \right\}}
\newcommand{\floor}[1]{{\lfloor #1 \rfloor}}

\newcommand{\trace}{{\mathrm{Tr}}}

\newcommand{\Order}{{\mathrm{O}}}
\newcommand{\order}{{\mathrm{o}}}
\newcommand{\density}[1]{\ketbra{#1}{#1}}
\newcommand{\eqdef}{\stackrel{\mathrm{def}}{=}}

\newcommand{\expct}{{\mathbb E}}

\newcommand{\identity}{{\mathbb I}}

\newcommand{\tensor}{\otimes}

\newcommand{\xor}{\oplus}

\newcommand{\adjoint}{\dagger}

\newcommand{\suppress}[1]{}
\newcommand{\comment}[1]{}

\newcommand{\etal}{\emph{et al.\/}}

\newcommand{\rI}{{{\mathrm I}}}
\newcommand{\rH}{{{\mathrm H}}}
\newcommand{\rS}{{{\mathrm S}}}
\newcommand{\rF}{{{\mathrm F}}}
\newcommand{\fh}[2]{{\mathfrak h}\!\left( #1 \;,\; #2 \right) }

\newcommand{\eps}{\varepsilon}
\newcommand{\tM}{\tilde{M}}
\newcommand{\tP}{\tilde{P}}
\newcommand{\tPi}{\tilde{\Pi}}
\newcommand{\ic}{\mathrm{IC}}
\newcommand{\qic}{\mathrm{QIC}}
\newcommand{\tqic}{\mathrm{\tilde{Q}IC}}

\newcommand{\sA}{{{\mathsf A}}}
\newcommand{\sB}{{{\mathsf B}}}

\newcommand{\dyck}{\mbox{\textsc{Dyck}}}
\newcommand{\Index}{\mbox{\textsc{Index}}}
\newcommand{\AIndex}{{\sc Augmented Index}}
\newcommand{\Asc}{\mbox{\sc Ascension}}

\newcommand{\cX}{{{\mathcal X}}}
\newcommand{\cY}{{{\mathcal Y}}}
\newcommand{\cS}{{{\mathcal S}}}
\newcommand{\cA}{{{\mathcal A}}}
\newcommand{\cB}{{{\mathcal B}}}
\newcommand{\cH}{{{\mathcal H}}}
\newcommand{\cM}{{{\mathcal M}}}

\newcommand{\hX}{\hat{X}}
\newcommand{\hY}{\hat{Y}}
\newcommand{\hZ}{\hat{Z}}
\newcommand{\hK}{\hat{K}}
\newcommand{\hB}{\hat{B}}

\title{ {\bf The space complexity of recognizing well-parenthesized 
expressions in the streaming model: \\
the Index function revisited}\thanks{The results on quantum
communication in this article were presented at the 15th Workshop on
Quantum Information Processing, QIP 2012, Dec., 2011.}
}

\author{Rahul Jain\thanks{%
Centre for Quantum Technologies and 
Department of Computer Science,
S15 \#04-01, 3 Science Drive 2, National University of Singapore,
Singapore 117543.
Email: {\tt rahul@comp.nus.edu.sg}.
Work done in part while visiting
Institute for Quantum Computing, University of Waterloo.
This work is supported by the Singapore Ministry of Education
Tier 3 Grant and the Core Grants of the Center for Quantum
Technologies, Singapore.
}
\and
Ashwin Nayak\thanks{%
Department of Combinatorics and Optimization,
and Institute for Quantum Computing, University of Waterloo,
200 University Ave.\ W.,
Waterloo, ON, N2L 3G1, Canada.
Email: {\tt ashwin.nayak@uwaterloo.ca}.
Work done in part at
Perimeter Institute for Theoretical Physics,
and while visiting
Center for Quantum Technologies, National University of Singapore.
Research supported in part by NSERC Canada, CIFAR, an ERA (Ontario),
QuantumWorks, MITACS, and ARO (USA). Research at
Perimeter Institute is supported in part by the Government of Canada
through Industry Canada and by the Province of Ontario through MRI.}
}

\date{March 6, 2014}

\suppress{
ACM class:
E.4 CODING AND INFORMATION THEORY 
F.2.2 Nonnumerical Algorithms and Problems
MSC class:
94A17, 94A15, 81P45, 68P30
}

\begin{document}



\maketitle

\abstract{
We show an~$\Omega(\sqrt{n}/T)$ lower bound for the space required by
any unidirectional constant-error randomized~$T$-pass streaming algorithm
that recognizes whether an expression over two types of parenthesis 
is well-parenthesized. This proves a conjecture due to Magniez, Mathieu, 
and Nayak (2009) and rigorously establishes that bidirectional 
streams are exponentially more efficient in space usage as compared with
unidirectional ones.
We obtain the lower bound by analyzing the information that is
necessarily revealed by the players about their respective inputs 
in a two-party communication protocol for a variant of the Index
function, namely Augmented Index. We show that in any communication protocol
that computes this function correctly with constant error on the uniform 
distribution (a ``hard'' distribution), either Alice reveals~$\Omega(n)$
information about her $n$-bit input, or Bob reveals~$\Omega(1)$ information
about his $(\log n)$-bit input, even when the inputs are drawn from an
``easy'' distribution, the uniform distribution over inputs which
evaluate to~$0$. The information cost trade-off is obtained by a 
novel application of the conceptually simple and familiar ideas 
such as \emph{average encoding\/} and the \emph{cut-and-paste 
property\/} of randomized protocols.

Motivated by recent examples of exponential savings in space by
streaming \emph{quantum\/} algorithms, we also study quantum protocols
for Augmented Index. Defining an appropriate notion of information cost 
for quantum protocols involves a delicate balancing act between its 
applicability and the ease with which we can analyze it. We define a 
notion of quantum information cost which reflects some of the
non-intuitive properties of quantum information. We show that in quantum 
protocols that compute the Augmented Index function correctly with
constant error on the uniform distribution, either Alice reveals~$\Omega(n/t)$
information about her $n$-bit input, or Bob reveals~$\Omega(1/t)$
information about his $(\log n)$-bit input, where~$t$ is the number of
messages in the protocol, even when the inputs are drawn from the
abovementioned easy distribution. While this trade-off demonstrates 
the strength of our proof techniques, it does not lead 
to a space lower bound for checking parentheses. We leave such an 
implication for quantum streaming algorithms as an intriguing open question.
}



\section{Introduction}
\label{sec-introduction}

Streaming algorithms~\cite{Muthukrishnan05} are designed to process 
massive input data, which cannot 
fit entirely in computer memory. Random access to such input is prohibitive,
so ideally we would like to process it with a single sequential scan.
Furthermore, during the computation, the algorithms are compelled to 
use space that is much smaller than the length of the input. Formally, 
streaming algorithms access the input sequentially,
one symbol at a time, a small number of times (called passes), while
attempting to solve some information processing task using as little
space (and time) as possible.
\suppress{
We refer the reader to the
text~\cite{Muthukrishnan05} for a more thorough introduction to this topic.
}

One-pass streaming algorithms that use constant space and time recognize
precisely the set of regular languages. It is thus natural to ask what
the complexity of languages higher up in the Chomsky hierarchy is in the
streaming model. In this work, we focus on a concrete such problem, that
of checking whether an expression with different types of parenthesis 
is well-formed.  
The problem is formalized through the language~$\dyck(2)$, which 
consists of all well-parenthesized expressions
over two types of parenthesis, denoted below by~$a, \overline{a}$
and~$b,\overline{b}$, with the bar indicating a closing parenthesis.
Formally,
$\dyck(2)$ is the language over alphabet~$\Sigma = \set{a, \overline{a},
b,\overline{b}}$ defined recursively as
\[
\dyck(2)
    \quad = \quad \epsilon + \bigl( a \cdot \dyck(2) \cdot \overline{a}
                  + b \cdot \dyck(2) \cdot \overline{b} \bigr) 
                  \cdot \dyck(2) \enspace,
\]
where~$\epsilon$ is the empty string, `$\cdot$' indicates concatenation
of strings (or subsets thereof) and `$+$' denotes set union.
This deceptively simple language is in a certain precise sense complete
for the class of context-free languages~\cite{ChomskyS63}, and is
implicit in a myriad of information processing tasks.

There is a straightforward algorithm that recognizes $\dyck(2)$ with
logarithmic space, as we may run through all possible levels of nesting, and
check parentheses at the same level. While this scheme is highly
space-efficient, it may make~$\Omega(n)$ passes over the input in the
worst case, on instances of length~$n$.
It is not obvious if we can translate this scheme to a 
streaming algorithm with a small number of passes over the input.
By appealing to the communication complexity of the equality function,
we can deduce that any \emph{deterministic\/} streaming algorithm for 
$\dyck(2)$ that makes~$T$ passes over the input requires 
space~$\Omega(n/T)$ on instances of length~$n$. Therefore, any streaming 
algorithm with smaller space complexity, if one exists, would necessarily 
be randomized.  One such algorithm is
suggested by a small-space algorithm for the word problem in the free 
group with~$2$ generators. This is a relaxation 
of $\dyck(2)$ in which local
simplifications~$\bar{p} p = \epsilon$ are allowed in addition to~$p
\bar{p} = \epsilon$ for every type of parenthesis~$(p, \bar{p})$.
There is a logarithmic space (randomized) algorithm for solving the word 
problem~\cite{LiptonZ77} that can easily be massaged into a one-pass
streaming algorithm with polylogarithmic space. Again, this algorithm
does not extend to~$\dyck(2)$.

We rigorously establish the impossibility of recognizing $\dyck(2)$ with
logarithmic space with a small number of passes in the streaming model,
even with randomized algorithms.
\begin{theorem}
\label{thm-1}
For any~$T \ge 1$,
any unidirectional randomized~$T$-pass streaming algorithm that 
recognizes length~$n$ instances of~$\dyck(2)$ with a constant 
probability of error uses space~$\Omega(\sqrt{n}/T)$.
\end{theorem}
A more precise statement of this theorem is presented as
Corollary~\ref{thm-dyck} later in this article.
\suppress{
(Similarly, the theorems we state below are made more precise in later 
sections.)
}

$\dyck(2)$ was first studied in the context of the streaming model 
by Magniez, Mathieu, and Nayak~\cite{MagniezMN10}. They were motivated
by its practical relevance, e.g., its relationship to the processing 
of large XML files, and by the connection between formal language theory 
and complexity in the context of processing massive data.
They overcome the apparent
difficulties described above and present sublinear space randomized 
streaming algorithms for~$\dyck(2)$. The first makes \emph{one\/} pass 
over the input, recognizes well-parenthesized expressions with
space~$\Order(\sqrt{n\log n}\,)$ bits, and has polynomially small
probability of error. Moreover, they prove that this one-pass algorithm
is optimal.  They establish that any one-pass randomized algorithm that makes 
error at most~$1/n\log n$ uses space~$\Omega(\sqrt{n\log n})$.
Theorem~\ref{thm-1} establishes a similar result for \emph{multi-pass\/}
streaming algorithms. The bound for one-pass algorithms given by
Theorem~\ref{thm-1} is a factor 
of~$\sqrt{\log n}$ better than the one in Ref.~\cite{MagniezMN10} for
constant error probability, but falls short of optimal (by the same
factor) for polynomially small error.

In the standard model for streaming algorithms, access to the input symbols 
is provided in the \emph{same fixed order in every pass over the input\/}.
This reflects a constraint of the infrastructure available to us in
practice. Theorem~\ref{thm-1}
applies to such \emph{unidirectional\/} algorithms. Perhaps surprisingly,
Magniez \etal{} showed that the demand on space shrinks drastically when 
algorithms for $\dyck(2)$ are allowed another pass over the input in the
\emph{reverse\/} direction.  They presented a second algorithm that makes
two passes in opposite directions over the input, uses
only~$\Order(\log^2 n)$ space, and has polynomially small probability of 
error.  A question that naturally arose is whether this is an artefact of
the algorithm, or if we could achieve similar reduction in space usage
by making multiple passes in the same direction. Magniez \etal{} conjecture
that a bound similar to that for the one-pass algorithms hold for 
multi-pass streaming algorithms if all passes are made in the same direction.
Theorem~\ref{thm-1} proves this conjecture and establishes the first 
natural example for which unidirectional multi-pass streaming algorithms are
much less powerful than bidirectional ones. More importantly, existing 
computing infrastructure only supports unidirectional streams, and this 
result confirms that we cannot reproduce the performance of the bidirectional
algorithm within it.
\suppress{
confirms the intuition
that the ability to scan the input in either direction gives streaming
algorithms additional computational firepower.
A relatively straightforward generalization of
the one-pass algorithm in Ref.~\cite{MagniezMN10} gives us a
unidirectional randomized~$T$-pass streaming algorithm that uses
space~$\Order\bigl(\sqrt{(n \log n)/T}\,\bigr)$ and has polynomially small
probability of error. The lower bound we derive thus comes within a
factor~$\sqrt{\log n} /T^{1/2}$ of optimal.
}

Theorem~\ref{thm-1} is a consequence of a lower bound that we establish 
for the ``information cost'' of two-party communication protocols for a
variant of the $\Index$ problem. In the $\Index$ problem, one
party, Alice, is given an~$n$-bit string~$x$, and the other party, Bob,
is given an integer~$k \in [n]$. Their goal is to determine the
bit~$x_k$ by communicating with each other.  In the variant we study, the
player holding the index also receives a portion of the other party's
input.  More formally, Alice holds an~$n$-bit string~$x$, and Bob, holds
an integer~$k \in [n]$, the prefix~$x[1,k-1]$ of~$x$, and a bit~$b \in
\set{0,1}$. The goal is to compute the function~$f_n(x,(k, x[1,k-1],
b)) = x_k \xor b$, i.e., to determine whether~$b = x_k$ or not.
This problem was studied in the one-way communication model, with
communication from Alice to Bob, as
``serial encoding''~\cite{AmbainisNTV99,Nayak99}. Lower bounds
on its quantum communication complexity were derived and used to
establish exponential lower bounds on the size of one-way quantum 
finite automata. In later works, the problem was studied as ``Augmented
Index''; the linear lower bound was re-derived for classical
communication, and used to establish lower bounds for streaming and
sketching (see, e.g.,~\cite{KaneNW10,DoBaIPW10}). 
The problem, called ``the Mountain problem'' by Magniez, Mathieu, and
Nayak~\cite{MagniezMN10}, was central to the proof of optimality of 
the one-pass streaming algorithm for \dyck(2). We elaborate on this
later in this section.

Informally speaking, we show that in any communication protocol that
computes the \AIndex\ function~$f_n$ with constant error on 
the uniform distribution~$\mu$ (a ``hard distribution''), either Alice
reveals~$\Omega(n)$ information about her $n$-bit input~$x$, or Bob
reveals~$\Omega(1)$ information about his $(\log n)$-bit input~$k$,
even when the inputs are drawn from an ``easy distribution'' ($\mu_0$, the
uniform distribution over~$f_n^{-1}(0)$). We formally define the notion of 
information cost~$(\ic_{\lambda}^\sA(\Pi),\ic_{\lambda}^\sB(\Pi))$ for a
protocol~$\Pi$ for the two players Alice ($\sA$) and Bob ($\sB$) with respect
to the distribution~$\lambda$ in Section~\ref{sec-main}, and show:
\begin{theorem}
\label{thm-2}
In any two-party randomized communication protocol~$\Pi$ for the
\AIndex\ function~$f_n$ that makes constant error at
most~$\eps \in [0, 1/4)$ on the uniform distribution~$\mu$ over inputs,
either~$\ic_{\mu_0}^\sA(\Pi) \in \Omega(n)$ or~$\ic_{\mu_0}^\sB(\Pi) \in
\Omega(1)$.
\end{theorem}
A more precise statement of this theorem is presented as
Theorem~\ref{thm-main} later in this article. We point out that
the theorem is optimal as there is a one-message deterministic protocol 
for \AIndex\ with communication~$n$.

The connection between streaming algorithms using ``small'' space to
two-party protocols for \AIndex\ with ``small'' information cost 
was presented by Magniez \etal{} for one-pass algorithms. However,
it generalizes in a straightforward manner to multi-pass algorithms.
For completeness, this reduction is described in full in 
Section~\ref{sec-streaming},
for multi-pass algorithms. The reduction consists of three steps,
following the information cost approach. (See, for example,
Refs.~\cite{ChakrabartiSWY01,SaksS02,Bar-YossefJKS04,JayramKS03,JainRS03b}
for earlier applications of this approach.)
First, a streaming algorithm for $\dyck(2)$ that uses space~$s$ is
mapped to a multi-party communication protocol in which the 
messages are each of the same length~$s$. Second, a two-party
communication protocol for \AIndex\ with ``small'' information cost 
with respect to~$\mu_0$ is derived using a ``direct sum'' argument.
Finally, a lower bound for the aforementioned information cost is proven.
Magniez \etal{} proved a lower bound for the information cost of 
a \emph{two-message\/} protocol
that resulted from a one-pass streaming algorithm. Our main contribution,
Theorem~\ref{thm-2}, lies in this final step. It
applies to protocols with an arbitrary number of messages, and is the
first general lower bound on information cost for \AIndex.

A notion of information cost for $\Index$ was studied previously by 
Jain, Radhakrishnan, and Sen~\cite{JainRS09} in the context of 
privacy in communication (see also earlier work due to
Klauck~\cite{Klauck04}). This notion differs from the one we study in
two crucial respects. First, it is defined in terms of the hard
distribution for the problem (uniform over all inputs). Second, the 
hard distribution is a product distribution.
The techniques they develop seem not to be directly relevant to the 
problem at hand, as we deal with an easy and non-product distribution. 

We devise a new method for analyzing the information cost of~$f_n$ to
arrive at Theorem~\ref{thm-2}.
The proof we present shows how conceptually simple and familiar ideas such as 
\emph{average encoding\/} and the \emph{cut-and-paste property\/} 
of randomized protocols may be brought to bear on \AIndex\
to derive the optimal (up to constant factors) information cost
trade-off. The intuition behind the lower bound is as follows.
Assume, for simplicity, that the protocol transcript contains the 
output. Starting from an input pair on which the function evaluates 
to~$0$, if the information cost of any one party is ``low'' and we 
carefully change her input, the transcript does not change ``much''.  We show 
that even when we simultaneously change the inputs with both parties, 
resulting in a~$1$-input of the function, the perturbation to the transcript
state is also correspondingly ``small''. This implies that
the two information costs cannot be ``small'' simultaneously.

We point out that the trade-off established by Magniez, Mathieu, and 
Nayak~\cite{MagniezMN10} for \emph{two-message\/} protocols that \emph{start 
with Alice\/}, and make polynomially small error, is stronger.
They show that either Alice reveals~$\Omega(n)$ information
about~$x$, or Bob reveals~$\Omega(\log n)$ information about~$k$ in such
protocols. This cannot be reproduced without a further refinement
of our techniques.  Indeed, Theorem~\ref{thm-2} also applies to 
two-message protocols in which \emph{Bob\/} starts. Such protocols 
match the trade-off given in the theorem: for every~$l \in 
\set{1,2, \ldots, \floor{\log_2 n}}$, there is a deterministic 
protocol for~$f_n$ in which Bob sends~$l$ bits of~$k$, and Alice 
responds with~$n/2^l$ bits.

In independent work, concurrent with ours, Chakrabarti, Cormode, 
Kondapally, and McGregor~\cite{ChakrabartiCKM10} derive a similar 
information cost trade-off for~$f_n$. Their motivation is identical to 
ours---to study the space required by unidirectional multi-pass 
streaming algorithms for $\dyck(2)$, and they present a similar space
lower bound for such algorithms. While some of the basic tools from 
information theory at the heart of their proof (e.g., the Chain Rule for
mutual information and the Pinsker Inequality) are equivalent to ours,
they take a different route to these tools. 
The first version of our article~\cite{JainN10a} and that of Chakrabarti
\etal~\cite{ChakrabartiCKM10a} contained trade-offs that were weaker,
albeit in different respects. After learning about each other's work,
both groups strengthened our respective proofs to achieve qualitatively
the same result. Subsequently, Chakrabarti and
Kondapally~\cite{ChakrabartiK11} extended the result to show that either
Bob reveals~$\Omega(b)$ information about his input~$k$, or Alice
reveals~$n/2^{\Order(b)}$ information about her input~$x$, i.e., 
either~$\ic_{\mu_0}^\sB(\Pi) \in \Omega(b)$ or~$\ic_{\mu_0}^\sA(\Pi) \in
n/2^{\Order(b)}$. This matches the information cost of the two-message
protocol described above up to constant factors.

The promise of fast processing with limited memory held by streaming 
algorithms make them especially attractive in the context of quantum 
computation. The absence of prototypes with a large enough number 
of qubits and long coherence times inevitably leads us to such
algorithms. This has fueled the study of quantum finite automata and 
also later works on quantum streaming 
algorithms~\cite{LeGall06,GavinskyKKRW08,Blume-KohoutCG14}. Several of
these works show how quantum effects lead to an exponential savings in 
space over their classical counterparts, albeit for specially crafted
problems. It is thus natural to ask how much more efficient such quantum 
algorithms could be, for a well-studied and important problem such as \dyck(2).
Motivated by this, we also study quantum protocols for \AIndex.
We define appropriate notions of quantum information 
cost~$(\qic_{\lambda}^\sA(\Pi), \qic_{\lambda}^\sB(\Pi))$ for
distributions~$\lambda$ with a limited form of dependence in 
Section~\ref{sec-qcommn}, and then arrive at the following trade-off.
\begin{theorem}
\label{thm-3}
In any two-party quantum communication protocol~$\Pi$ (with
read-only behaviour on inputs and no intermediate measurements) for the
\AIndex\ function~$f_n$ that has~$t$ message exchanges and 
makes constant error at
most~$\eps \in [0, 1/4)$ on the uniform distribution~$\mu$ over inputs,
either~$\qic_{\mu_0}^\sA(\Pi) \in \Omega(n/t)$ or~$\qic_{\mu_0}^\sB(\Pi) \in
\Omega(1/t)$.
\end{theorem}

Quantum protocols have the ability to compute without revealing much
information~\cite{GavinskyI13,FawziHS13}. It is thus hardly a surprise
that the quantum 
information cost trade-off involves a number of subtleties. For instance, it 
is not obvious how we may quantify information cost in the absence of 
the notion of a message transcript, or how we discount information 
leakage due to the non-product nature of the input distribution. These
issues are discussed in detail in Section~\ref{sec-qcommn}.
Nonetheless, we show how the ideas behind Theorem~\ref{thm-2}
also shed light on quantum communication. The intuition
from the classical case comes with its own complications, such as the
absence of an analogue of the Cut-and-Paste Lemma.
We circumvent the Cut-and-Paste property by appealing to the ``Local
Transition Theorem'' and adapting a hybrid argument due to Jain, 
Radhakrishnan, and Sen~\cite{JainRS03b}. We apply these  
on a message-by-message basis, which leads to the dependence of the 
trade-off on the number of messages in the protocol.
We are not aware of quantum protocols that beat the classical
information bounds. However the dependence of the trade-off in
Theorem~\ref{thm-3} on the number of messages~$t$ may be inherent,
as is the case  with Set Disjointness~\cite{JainRS03b}.

Theorem~\ref{thm-3} demonstrates the versatility of our proof techniques.
The techniques due to Magniez \etal{}~\cite{MagniezMN10} and Chakrabarti 
\etal{}~\cite{ChakrabartiCKM10} for showing information cost trade-off 
in classical protocols do not seem to generalize to quantum protocols.
They analyze the input distribution conditioned on the message 
transcript, a notion for which no suitable quantum analogue is known.
Theorem~\ref{thm-3}, however, does not immediately lead to a lower bound
on the space required by quantum streaming algorithms for \dyck(2).
The main hurdle here is that the connection between streaming 
algorithms and communication protocols for \AIndex\ with low information 
cost does not extend to the quantum case. This appears to be due to 
the stronger notion of information cost that we adopt. (The stronger
notion appears to be necessary for our proof technique.) It is possible
that a version of Theorem~\ref{thm-3} hold with an alternative
definition of information cost that is more relevant to quantum
streaming algorithms. We leave this for future investigation.

Communication problems involving the $\Index$ and \AIndex\ functions 
capture a number
of phenomena in the theory of computing, both classical and quantum,
in addition to playing a fundamental role in the area of communication
complexity~\cite{KushilevitzN97}. For instance, they have been used to analyze
data structures~\cite{MiltersenNSW98}, the size of finite
automata~\cite{AmbainisNTV02} and formulae~\cite{Klauck07},
the length of locally decodable codes~\cite{KerenidisW04},
learnability of states~\cite{KremerNR99,Aaronson07}, and sketching 
complexity~\cite{Bar-YossefKK04}. Recently, phenomena in quantum information
have been discovered via the $\Index$ function problem, e.g., information
causality~\cite{PawowskiPKSWZ09}, a connection between non-locality
and the uncertainty principle~\cite{OppenheimW10} and quantum
ignorance~\cite{VidickW11}. We believe that the more nuanced properties
of the \AIndex\ function such as the one we establish here are of fundamental
importance, and are likely to find application in other contexts as well.

\subsection*{Acknowledgments}

We thank Fr{\'e}d{\'e}ric Magniez and Christian Konrad for their comments 
on an earlier version of this article.
A.N.\ thanks Fr{\'e}d{\'e}ric Magniez also for several helpful discussions
preceding this work.

We thank the authors of Ref.~\cite{ChakrabartiCKM10} for sending us
their initial manuscript when we first publicized an earlier version
of the article. The (classical) results in our
respective articles were originally weaker in incomparable ways, and
the exchange inspired both groups to refine our analyses to obtain the
current classical information cost trade-off results.

We are grateful to the anonymous referees for their help in improving 
the presentation.

\section{Classical information cost of Augmented Index}
\label{sec-classical}

In this section we present the first result of this article.  We
summarize the notational conventions we follow and the background from
classical information theory that we assume in
Section~\ref{sec-info-theory}. We do the same for two-party
communication complexity and information cost in Section~\ref{sec-commn}.
Then we develop the lower bound for
classical protocols for \AIndex\ in Section~\ref{sec-main}.

\subsection{Information theory basics}
\label{sec-info-theory}

We reserve small case letters like~$x,k,m$ for bit-strings or
integers, and capital letters like~$X,K,M$ for random variables over
the corresponding sample spaces. We use the same symbol for a random
variable and its distribution. As is standard, given jointly
distributed random variables $AB$ over a product sample space, $A$
represents the marginal distribution over the first component. We
sometimes use $A|b$ as shorthand for the conditional
distribution~$A|(B=b)$ when the second random variable~$B$ is clear
from the context. For a string $x \in \{0,1\}^n$, and integers~$i,j
\in [n]$, where~$[n] = \set{1,2, \ldots, n}$, we let $x[i, j]$ denote the substring
of consecutive bits~$x_i \dotsm x_j$. If~$j < i$, the expression
denotes the empty string. This notation extends to random variables
over~$\set{0,1}^n$ in the obvious manner.  When a sample~$z$ is drawn
from distribution~$Z$, we denote it as~$z \leftarrow Z$.  

The~$\ell_1$ distance~$\norm{A-B}$ between two random variables~$A,B$
over the same finite sample space~$\cS$ is given by
\[
\norm{A-B} \quad = \quad \sum_{i \in \cS} \size{A(i) - B(i)}\enspace,
\]
and takes values in the interval~$[0,2]$.
(Recall that as per our notational convention~$A(i), B(i)$ denote the
probabilities assigned to~$i \in \cS$ by~$A,B$, respectively.)  The
Hellinger distance~$\fh{A}{B}$ between the random variables is defined
as
\[
\fh{A}{B} \quad = \quad \left[ \frac{1}{2} \sum_{i \in \cS} 
        \left( \sqrt{A(i)} - \sqrt{B(i)} \right)^2 \right]^{1/2}
        \enspace.
\]

Hellinger distance is a metric, and is related to~$\ell_1$ distance
in the following manner. (See Section~3.2 in~\cite{LeCamY90} for a
proof.)
\comment{The property is stated and proven on page 25 of Le Cam and
Yang, 1990.}
\begin{proposition}
\label{fact:ellhell} 
Let $P,Q$ be distributions over the same sample space. Then
\[
\fh{P}{Q}^2 
    \quad \leq \quad \frac{1}{2} \norm{P-Q} 
    \quad \leq \quad \sqrt{2} \; \fh{P}{Q} \enspace.
\]
\end{proposition}

The square of the Hellinger distance satisfies the following property,
called \emph{joint convexity\/}. It may be verified by a straightforward
application of the Cauchy-Schwarz inequality.
\begin{proposition} 
\label{fact-hconvex} 
Let $P_i,Q_i$ be distributions over the same sample space for each~$i
\in [n]$, and let~$(\alpha_i)$ be a probability distribution
over~$[n]$.  Let~$P = \sum_{i = 1}^n \alpha_i P_i$, and~$Q = \sum_{i =
  1}^n \alpha_i Q_i$. Then
\[
\fh{P}{Q}^2 \quad \leq \quad \sum_{i = 1}^n \alpha_i\, \fh{P_i}{Q_i}^2
\enspace.
\]
\end{proposition}
\begin{proof}
By the Cauchy-Schwarz Inequality, for each~$j \in \cS$,
\begin{eqnarray*}
\sqrt{P(j) \, Q(j)} 
    & = & \left[ \left( \sum_{i \in [n]} \alpha_i \, P_i(j) \right)
          \left( \sum_{i' \in [n]} \alpha_{i'} \, Q_{i'}(j) \right)
          \right]^{1/2} \\
    & \ge & \sum_{i \in [n]} \sqrt{\alpha_i \, P_i(j)} 
            \sqrt{\alpha_i \, Q_i(j)} \enspace.
\end{eqnarray*}
So we have
\begin{eqnarray*}
\fh{P}{Q}^2
    & = & \frac{1}{2} \sum_{j \in \cS}
          \left( P(j) + Q(j) - 2 \sqrt{P(j) \, Q(j)} \right) \\
    & \le & \frac{1}{2} \sum_{j \in \cS} \sum_{i \in [n]} \alpha_i
            \left( P_i(j) + Q_i(j) - 2 \sqrt{P_i(j) \, Q_i(j)}
            \right) \\
    & = & \sum_{i = 1}^n \alpha_i\, \fh{P_i}{Q_i}^2 \enspace.
\end{eqnarray*}
\end{proof}

We rely on a number of standard results from information theory in
this work.  For a comprehensive introduction to the subject, we refer
the reader to a text such as~\cite{CoverT91}.

We use~$\rH(X)$ to denote the Shannon entropy of the random variable~$X$,
$\rI(X:Y)$ to denote the mutual information between two random
variables~$X,Y$, and~$\rI(X:Y|Z)$ to denote the conditional mutual 
information of~$X,Y$ with respect to a jointly distributed random 
variable~$Z$.  We also use~$\rH(p)$ to denote the Binary entropy
function when~$p \in [0,1]$.

The chain rule for mutual information, Theorem~2.5.2 in~\cite{CoverT91},
states:
\begin{proposition}[Chain Rule] 
\label{fact-chain}
Let $ABC$ be jointly distributed random variables. Then 
$$ \rI(AB:C)  \quad = \quad  \rI(A:C) + \rI(B : C \,|\, A)  \enspace . $$
This implies that for jointly distributed random variables $A_1 \dotsb A_n C$,
$$ \rI(A_1 \dotsb A_n:C)
     \quad = \quad  \rI(A_1:C) + \rI(A_2 : C \,|\, A_1)
                    + \cdots +  \rI(A_n : C \,|\, A_1 \cdots A_{n-1})
                    \enspace .
$$ 
\end{proposition}

The Average encoding theorem~\cite{KlauckNTZ07,JainRS03b} is a
quantitative version of the intuition that two random variables that
are only weakly correlated are nearly independent. Stated differently,
the conditional distribution of one given the other is close to its
marginal distribution, if their mutual information is sufficiently
small.
\begin{proposition}[Average encoding theorem~\cite{KlauckNTZ07,JainRS03b}]
\label{thm-avg} 
Let $AB$ be jointly distributed random  variables.  Then,
\[
\expct_{b \leftarrow B} \; \fh{A|b}{A}^2 \quad \leq \quad \kappa \,
\rI(A:B) \enspace,
\]
where~$\kappa$ is the constant~$\frac{\ln 2}{2}$.
\end{proposition}

\subsection{Communication protocols and information cost}
\label{sec-commn}

In the two-party communication model~\cite{Yao79} for computing Boolean
functions, parties Alice and Bob receive inputs~$x \in \cX$ 
and~$y \in \cY$, respectively, for some sets~$\cX, \cY$. They may
share a random bit string~$R$, that is independent of
the inputs~$x,y$. The bits of~$R$ are called \emph{public\/} coins, as
they are known to both parties. Alice (or Bob) may use an additional random
string~$R_\sA$ ($R_\sB$, respectively), that is not known to the
other party. These strings~$R_\sA, R_\sB$ are called \emph{private\/} coins.

The goal of the two parties is to compute a bi-variate Boolean 
function~$f : \cX \times \cY \rightarrow \set{0,1}$, by communicating 
with each other. The communication occurs in the form of~$t \ge 0$ messages,
starting with one party, and then alternating with the other. In each of
the~$t$ steps,
the party sending it computes the message as a function of the input, the 
public and private random coins she or
he has, and the messages received so far. After
all~$t$ messages have been sent, the recipient of the last message produces
the output of the protocol. The output is computed in a manner analogous
to the messages, from the party's input, random coins, and all the messages
received.

The pattern of communication is specified by a \emph{protocol\/}~$\Pi$,
which lists the type, number, and distribution of the coins used by 
each party, the
number of messages, the party that starts the protocol, and the
functions used by the parties to generate the messages and the output. 
The sequence of~$t$ messages produced during a run of the protocol~$\Pi$ on
a pair of inputs~$x,y$ together constitute the \emph{transcript\/}. This
is in general a random variable due to the use of random coins. We
denote the random variable corresponding to the output by~$\Pi(x,y)$.
We point out that the transcript need not include the output of the protocol.

The probability of correctness (or \emph{success\/}) of a protocol on 
input~$x,y$ is~$\Pr[\Pi(x,y) = f(x,y)]$. We consider inputs drawn 
from a joint distribution~$XY$, in which case the success probability 
is~$\Pr[\Pi(X,Y) = f(X,Y)]$. The probability of the complementary event
is called the \emph{error\/} of the protocol on the distribution~$XY$.

We refer the reader to the text~\cite{KushilevitzN97} for
equivalent formulations of communication protocols, and a thorough
introduction to the models of two-party classical communication.

Protocols that use only public coins are called public-coin protocols
and those that use only private coins are called private-coin protocols.
The availability of public randomness obviates the need for private 
randomness in typical settings. Conversely, private randomness can 
often simulate public coins with a slight increase in 
communication~\cite{Newman91}. In the context of information cost,
however, access to the private randomness used by one party may result 
in more information being revealed to the other. To the
best of our knowledge, there is no general recipe for replacing private 
with public randomness while preserving information cost. (For recent
progress on this question, see Ref.~\cite{BrodyBKLS13}.) In the
reductions between protocols we encounter in this article, regardless 
of the nature of randomness used in the original protocol, we end up 
with a protocol with both types of randomness. We therefore study protocols 
of this type.

We use the following Cut-and-Paste property of private-coin
communication protocols. (For a proof, see Lemma~6.3 in
Ref.~\cite{Bar-YossefJKS04}.)

\begin{proposition}[Cut-and-Paste~\cite{Bar-YossefJKS04}]
\label{fact:cutandpaste}
Let $\Pi$ be a two-party private-coin communication protocol. Let
$M(x,y)$ denote the random variable representing the message transcript
in $\Pi$ when the first party has input $x$ and the second party
has input $y$. Then for all pairs of inputs~$(x,y)$ and~$(u,v)$,
\[
\fh{M(x,y)}{M(u,v)} \quad = \quad \fh{M(x,v)}{M(u,y)}
\enspace .
\]
\end{proposition}

We consider the information revealed during a communication
protocol and focus on a notion known as ``internal information'' in
the literature. Although this notion is implicit in earlier
work~\cite{Bar-YossefJKS04}, it was named so by Barak, Braverman, Chen,
and Rao~\cite{BarakBCR13}. We emphasize that there is no canonical
measure of information cost, and the choice of definition is
often driven by a motivating application. A different
definition of information cost would suffice for our application 
to streaming algorithms, and would additionally simplify some of our
proofs. However, we use internal information, as this
gives us the strongest information cost trade-off result.

Consider a randomized two-party communication protocol~$\Pi$ which
uses public randomness~$R$, and may additionally use private
randomness.  Suppose that~$M$ is the message transcript of the
protocol, when the inputs to the two players, Alice and Bob,
respectively, are sampled from the joint distribution~$\lambda$. Let the
input random variables be denoted by~$X,Y$. The
\emph{information cost\/} of the protocol for Alice with respect to
the distribution~$\lambda$ is defined as~$\ic^\sA_\lambda(\Pi) \eqdef
\rI(X:M \,|\, YR)$.  The information cost of the protocol for Bob is
defined symmetrically as~$\ic^\sB_\lambda(\Pi) \eqdef \rI(Y : M \,|\,
XR)$. These quantities measure the amount of information about one
party's input that the other gains through the course of the protocol.

Note that we could have conditioned on the private randomness used by one
party (say, Bob) as well in the other's (Alice's) information cost. This
is however redundant, as given his input~$Y$, the public randomness~$R$,
and the message transcript~$M$, Bob's private randomness~$R_\sB$
is independent of Alice's input (and private randomness). Indeed, by the 
Chain Rule (Proposition~\ref{fact-chain}),
\begin{eqnarray*}
\rI(X:M \,|\, YR R_\sB)
    & = & \rI(X:M R_\sB \,|\, YR) - \rI(X: R_\sB \,|\, YR) \\
    & = & \rI(X:M R_\sB \,|\, YR) \\
    & = & \rI(X: M \,|\, YR) + \rI(X: R_\sB \,|\, YRM) \\
    & = & \rI(X: M \,|\, YR) \enspace.
\end{eqnarray*}

\subsection{The classical information cost lower bound}
\label{sec-main}

The first main theorem in this article may be viewed as a trade-off 
between information revealed by the two parties about their inputs while
computing the \AIndex\ function~$f_n$. We show that at least
one of the parties necessarily reveals ``a lot'' of information even
on an ``easy distribution'' if the protocol computes~$f_n$ with
bounded error on a ``hard distribution''. 

Recall that in the \AIndex\ problem, one party, Alice, has
an~$n$-bit string~$x$, and the other party, Bob, has an integer~$k \in
[n]$, the prefix~$x[1,k-1]$ of~$x$, and a bit~$b \in \set{0,1}$. Their
goal is to compute the function~$f_n(x,(k, x[1,k-1], b)) = x_k \xor
b$, i.e., to determine whether~$b = x_k$ or not, by engaging in a
two-party communication protocol.

Let $(X,K,B)$ be random variables distributed according to~$\mu$, the
uniform distribution over~$\{0,1\}^n \times [n] \times
\set{0,1}$. Let~$\mu_0$ denote the distribution conditioned upon~$B =
X_K$, i.e., when the inputs are chosen uniformly from the set of~$0$s
of~$f_n$. We are interested in the information cost of a protocol~$\Pi$
with public randomness~$R$ for \AIndex\ under the
distribution~$\mu_0$, for the two parties. Let~$M$ denote the entire
message transcript under~$\mu$, and let~$M^0$ denote the transcript
under distribution~$\mu_0$. Then the information cost of~$\Pi$ is given
by~$\ic^\sA_{\mu_0}(\Pi) = \rI(X:M^0 \,|\, X[1,K]\, R)$ 
and~$\ic^\sB_{\mu_0}(\Pi) = \rI(K:M^0 \,|\, XR)$. Note that~$X[1,K]=X[1,K-1]B$
under distribution $\mu_0$ and that~$K$ can be computed from~$X[1,K]$.
Hence~$K,B$ are not explicitly included in Bob's input in the expression for 
Alice's information cost. Similarly, $X[1,K-1] \, B$ are determined by~$K$
when we condition on~$X$ under distribution~$\mu_0$. Hence, these are not
explicitly included in Bob's input in the expression for his information
cost. The use of the
notation~$M^0$ is equivalent to conditioning on the event~$X_K = B$,
i.e., imposing the distribution~$\mu_0$, and helps us present our
arguments more cleanly. 

Since the value of the \AIndex\ function~$f_n$ is a constant
on~$\mu_0$, there is no \emph{a priori\/} reason for the information
cost of any party in a protocol to be large. However, we additionally
require the protocol to be correct with non-trivial probability on the
uniform distribution, under which there is equal chance of the
function being~$0$ or~$1$. If the information cost (under~$\mu_0$) of
the two parties is sufficiently low, we show that neither party can
determine with high enough confidence what the function value is. The
intuition behind this is as follows. Suppose we restrict the inputs
to~$\mu_0$. If Bob's input~$K$ is changed, the random variables in
Alice's possession, specifically the message transcript~$M^0$
conditioned on her inputs, are not perturbed by much. This is because
these random variables reveal little information about~$K$.
Similarly, if we flip one of the bits of Alice's input~$X$ outside 
of the prefix with Bob, the random variables in Bob's possession at 
the end of the protocol are not perturbed by much.  Formally, these 
properties follow from the Average Encoding Theorem. Observe that if 
we simultaneously change Bob's index~$K$ to some~$L > K$ and flip 
the~$L$th bit of~$X$, we switch from a~$0$-input of~$f_n$ to a~$1$-input.
The Cut-and-Paste Lemma ensures that by simultaneously changing the 
inputs with the two parties, the message transcript is perturbed by at 
most the sum of the amounts when
the inputs are changed one at a time. This implies that the message
transcript does not sufficiently help either party compute the
function value.

We formalize this intuition in the next theorem, which we state for 
even~$n$. A similar result holds for odd~$n$, and may be derived 
from the proof for the even case. Together, they give us
Theorem~\ref{thm-2}, as stated in the introduction
(Section~\ref{sec-introduction}).
\begin{theorem}
\label{thm-main}
For any two-party randomized communication protocol~$\Pi$ for the
\AIndex\ function~$f_n$ with~$n$ even, that makes error at
most~$\eps \in [0, 1/4)$ on the uniform distribution~$\mu$ over
inputs, we have
\[
\left[ \frac{\ic_{\mu_0}^\sA(\Pi)}{n} \right]^{1/2}
+  \left[ 2 \cdot \ic_{\mu_0}^\sB(\Pi) \right]^{1/2}
\quad \geq \quad \frac{1-4\eps}{4\sqrt{\ln 2}}
                 - \left[ \frac{\rH(2\eps)}{n} \right]^{1/2} \enspace,
\]
where~$\mu_0$ is the uniform distribution over~$f_n^{-1}(0)$. In
particular, for any~$\eps$ smaller than~$1/4$ by a constant,
either~$\ic_{\mu_0}^\sA(\Pi) \in \Omega(n)$ or~$\ic_{\mu_0}^\sB(\Pi) \in
\Omega(1)$.
\end{theorem}
\begin{proof}
Consider a protocol~$\Pi$ as in the statement of the theorem. 
Let the inputs be given by random variables~$X,K,B$, drawn from
the distribution~$\mu$.

Let~$M$ be the entire message transcript of the protocol, and let~$M^0$
be the transcript under distribution~$\mu_0$. Without
loss of generality, we assume that Bob computes the output of the
protocol. If Alice computes the output, we include an additional message
from her to Bob consisting of the output. We show below that this only 
marginally increases the information revealed by Alice, and include its
effect in the lower bound we derive. Indeed, if the single
bit output of the protocol is~$O^0$ under the distribution~$\mu_0$,
$\rH(O^0) \leq \rH(2\eps)$, as the protocol produces the correct
output with probability at least~$1-2\eps$ on the
distribution~$\mu_0$. Let~$d \ge 0$ be such 
that~$\rI(X: M^0 \,|\, X[1,K]) = dn$. Then,
\begin{eqnarray*}
\rI(X: M^0 O^0 \,|\, X[1,K])
    & = & \rI(X: M^0 \,|\, X[1,K]) + \rI(X: O^0 \,|\, M^0 X[1,K]) \\
    & \leq & dn + \rH(O^0) \enspace,
\end{eqnarray*}
and~$\rI(K : M^0 O^0 \,|\, X) = \rI(K : M^0 \,|\, X)$.
Henceforth, we assume that the output of the protocol~$\Pi$
is computed by Bob, and its information costs are bounded
as~$\ic_{\mu_0}^\sA(\Pi) \leq d_1 n$ with~$d_1 = d + \rH(2\eps)/n$, 
and~$\ic_{\mu_0}^\sB(\Pi) \leq c$.

Let~$R$ be the public randomness used in the protocol. For each 
specific value~$r$ for the public random coins, we use the subscript~$r$ 
on a random variable to denote conditioning on~$R = r$. In particular, the
random variable~$M^0_r$ is the transcript~$M$ conditioned on~$R = r$, 
under distribution~$\mu_0$. Define~$d_{1r} 
\eqdef \tfrac{1}{n} \; \rI(X:M^0_r \,|\, X[1,K])$ and~$c_r \eqdef
\rI(K:M^0_r \,|\, X)$, so
that~$\expct_{r \leftarrow R} \; d_{1r} = \ic^\sA_{\mu_0}(\Pi) / n$
and~$\expct_{r \leftarrow R} \; c_r = \ic^\sB_{\mu_0}(\Pi)$.
We emphasize that the protocol may use private randomness in addition
to the public randomness~$R$. Let~$\eps_r$ denote the error made by the
protocol~$\Pi$ on the uniform distribution~$\mu$ over inputs, when~$R =
r$.

In the rest of the proof, we fix a specific value~$r$ for the public
randomness, and show that
\begin{equation}
\label{eqn-thm}
d_{1r}^{1/2} +  (2 c_r)^{1/2}
    \quad \geq \quad \frac{1-4\eps_r}{4\sqrt{\ln 2}} \enspace.
\end{equation}
Averaging this over~$r \leftarrow R$ and applying the Jensen Inequality 
gives us the theorem.

We show below that the random variables~$M_r^0 X[1,K]$ with Bob are
``close'' in distribution to the random
variables~$M_r^1 X[1,K-1]\, \bar{X}_K$, where~$M_r^1$
denotes the transcript~$M_r$ conditioned on the function value being~$1$,
i.e., when~$B = {\bar{X}_K}$. In other words, we show that
the~$\ell_1$ distance between them is only ``slightly more'' than~$1$ if the 
information cost of the protocol is small.
\begin{lemma}
\label{thm-close}
$\norm{M_r^0  X[1,K] - M_r^1 X[1,K-1]\, \bar{X}_K} 
    \quad \leq \quad 1 + 8 \sqrt{\kappa\, c_r} + 4 \sqrt{2 \kappa\, d_{1r}} $,
where~$\kappa = \frac{\ln 2}{2}$.
\end{lemma}

For any fixed~$r$, given the message transcript and his input, Bob's 
private randomness is independent of Alice's input and private randomness.
Therefore, we can regenerate Bob's private randomness exactly from the 
other random variables in his possession. As a result, we may use 
the protocol~$\Pi$ to identify the two distributions,
$M_r^0 X[1,K]$ and~$M_r^1 X[1,K-1]\, \bar{X}_K$, with average error~$\eps_r$.
If the error~$\eps_r$ were small, the~$\ell_1$ distance would be
correspondingly closer to~$2$. Formally, the $\ell_1$ distance between 
two distributions is non-increasing under the action of a stochastic map. 
So $\norm{M_r^0  X[1,K] - M_r^1 X[1,K-1]\, \bar{X}_K} \geq 2(1-2\eps_r)$,
as the latter is a lower bound on the~$\ell_1$ distance between the 
distributions of the output of the protocol in the two cases.
This gives us a lower bound on the information cost, in terms of the error
made by the protocol.  Combining the two bounds on the~$\ell_1$ 
distance, we get Eq.~(\ref{eqn-thm}) and hence the theorem.
\end{proof}

We now prove the heart of the theorem, i.e., that the message transcript
for the~$0$ and~$1$ inputs are close to each other in distribution.

\begin{proofof}{Lemma~\ref{thm-close}} 
{}
The proof follows the intuition given before Theorem~\ref{thm-main}.
We break the proof into several steps, each of which is captured by a
lemma. The proofs of the lemmata are postponed to later in the section
so as to present the high-level argument first.

When we wish to explicitly write the transcript~$M_r$ as a function of
the inputs to Alice and Bob, say $x$ and $x[1,k-1],b$ respectively, we
write it as $M_r(x ; x[1,k-1],b)$. If~$b = x_k$, we write Bob's input
as~$x[1,k]$.

For any~$x \in \set{0,1}^n$ and~$i \in [n]$, let~$x^{(i)}$ denote the
string that equals~$x$ in all coordinates except at the~$i$th.
Since~$(X,X[1,K-1],\bar{X}_K)$ and~$(X^{(K)},X[1,K])$ are identically
distributed, $M_r^1 = M_r(X ; X[1,K-1],\bar{X}_K)$ has the same distribution
as~$M_r( X^{(K)} ; X[1,K])$. Thus, our goal is to bound
\[
\norm{ M_r(X ; X[1,K]) \, X[1,K] - M_r( X^{(K)} ; X[1,K]) \, X[1,K] }
\enspace.
\]

Later, we consider the random variables in Bob's possession when we flip
one of the bits in input~$X$ with Alice. In order to do the flip in a manner
consistent with the prefix with Bob, we only flip bits in coordinates~$>
n/2$. This gives us a bound on the above quantity when the index 
is larger than~$n/2$. Therefore we consider~$L$ uniformly and independently 
distributed in~$[n]-[n/2]$, and~$J$ be uniformly and independently 
distributed in~$[n/2]$.  We have
\begin{eqnarray}
\nonumber
\lefteqn{ \norm{ M_r(X ; X[1,K]) \, X[1,K]
                 - M_r( X^{(K)} ; X[1,K])  \, X[1,K]}}  \\
\label{eqn-bound}
    & = &  \left\|  \frac{1}{2} \left( M_r(X ; X[1,J]) \, X[1,J] + M_r(X ; X[1,L]) \, X[1,L] \right) \right. \nonumber \\
    &  & \left. \mbox{} -   \frac{1}{2} \left( M_r( X^{(J)} ; X[1,J])  \, X[1,J] + M_r( X^{(L)} ; X[1,L])  \, X[1,L]  \right)\right\|  \nonumber \\
    & \leq & \frac{1}{2}  \norm{ M_r(X ; X[1,J]) \, X[1,J] - M_r( X^{(J)} ; X[1,J])  \, X[1,J] } \nonumber \\
    & & \mbox{} +  \frac{1}{2}  \norm{M_r(X ; X[1,L]) \, X[1,L] - M_r( X^{(L)} ; X[1,L])  \, X[1,L]} \nonumber \\
            & \leq & 1 + \frac{1}{2} 
             \norm{ M_r(X ; X[1,L]) \, X[1,L]
                    - M_r( X^{(L)} ; X[1,L]) \, X[1,L] }
             \enspace,
\end{eqnarray}
and we bound the RHS from above.

Recall that our goal is to show that, on average, changing from
a~$0$-input to a~$1$-input does not perturb the message transcript by much.
For this, we begin by showing that changing Alice's input alone, or
similarly, Bob's input alone, has this kind of effect.
If the information cost of Bob is small, the
message transcript does not carry much information about~$K$ when the
inputs are drawn from~$\mu_0$. From this, we deduce that the
transcript~$M_r^0$ is (on average) nearly the same for different inputs
to Bob.

 We compare the transcript when Bob's
input index is~$J$ to when it is~$L$.
\newtheorem*{thm-pairs}{Lemma~\ref{thm-pairs}}
\begin{thm-pairs}
$
\expct_{(x,j,l) \leftarrow (X,J,L)}
\; \fh{M_r(x \,;\, x[1,j])}{ M_r(x \,;\, x[1,l])}^2 
\quad \leq \quad 8 \kappa \, c_r.
$
\end{thm-pairs}
We defer the proof to later in this section.

In the interest of readability, we abbreviate some random variables in 
the rest of the proof, as also in the intermediate lemmata.
For~$i \in [n]$, and a prefix~$x[1,i]$ of a string~$x \in \set{0,1}^n$
that will be clear from the context, let~$v_i$ denote the 
prefix~$x[1,i]$, let~$U_i$ denote the random variable~$x[1,i]\, X[i+1,n]$
(i.e.,~$X$ conditioned on having prefix~$v_i$), and let~$U'_i$ denote 
the random variable~$x[1,i-1]\, \bar{x}_i\, X[i+1,n]$ (i.e.,~$U_i$ with 
the~$i$th bit flipped).

When changing Alice's input, we would like to ensure that the prefix
held by Bob does not change. So we restrict our
attention to Bob's inputs with index~$J \in [n/2]$, and change Alice's 
input by flipping the~$L$th bit, with~$L \in [n] - [n/2]$.
If the information cost of Alice is small, $M_r^0$ does not
carry much information about~$X$, even given a
prefix. Therefore, flipping a bit outside the prefix does not perturb the
transcript by much.
\newtheorem*{thm-flip}{Lemma~\ref{thm-flip}}
\begin{thm-flip}
$\expct_{(x[1,l],j,l) \leftarrow (X[1,L],J,L)} 
\; \fh{ M_r(U_l \,;\, v_j)}{ M_r(U'_l \,;\, v_j)}^2 
\quad \leq \quad 16 \kappa\, d_{1r} \enspace.$
\suppress{
We have
\begin{align*}
\; \fh{ M_r(x[1,l]\, X[l+1,n] \,;\, x[1,j])} 
{ M_r(x[1,l-1]\, \bar{x}_l\, X[l+1,n] \,;\, x[1,j])}^2 \\
    & \leq \quad 16 \kappa\, d_{1r} \enspace.
\end{align*}
}
\end{thm-flip}
This is proven later in the section.

We now conclude the proof of Lemma~\ref{thm-close}. Since Hellinger
distance squared is jointly convex (Proposition~\ref{fact-hconvex}), 
Lemma~\ref{thm-pairs} gives us a bound on the distance between the
transcripts averaged over the choice of suffix~$x[l+1,n]$. Along with
the Jensen Inequality, we get
\begin{align}
\label{eq:1}
\expct_{(x[1,l],j,l) \leftarrow (X[1,L],J,L)} 
\; \fh{M_r(U_l \,;\, v_j)}{ M_r(U_l \,;\, v_l)} 
    & \quad \leq \quad \sqrt{8 \kappa \, c_r} \enspace.
\suppress{
\expct_{(x[1,l],j,l) \leftarrow (X[1,L],J,L)} &
\; \fh{M_r(x[1,l]\, X[l+1,n] \,;\, x[1,j])}{
M_r(x[1,l]\, X[l+1,n] \,;\, x[1,l])} \nonumber \\
    & \leq \quad \sqrt{8 \kappa \, c_r} \enspace.
}
\end{align}
Along with the Triangle Inequality, and Lemma~\ref{thm-flip}, this
implies that
\begin{align*}
\expct_{(x[1,l],j,l) \leftarrow (X[1,L],J,L)} 
\; \fh{ M_r(U_l \,;\, v_l)}{ M_r(U'_l \,;\, v_j)} 
\quad \leq \quad \sqrt{8 \kappa\, c_r} + \sqrt{16 \kappa\, d_{1r}}\,
\enspace.
\suppress{
\expct_{(x[1,l],j,l) \leftarrow (X[1,L],J,L)} &
\; \fh{ M_r(x[1,l]\, X[l+1,n] \,;\, x[1,l])}{
M_r(x[1,l-1]\,\bar{x}_l\, X[l+1,n] \,;\, x[1,j])} \\
    & \leq \quad \sqrt{8 \kappa\, c_r} + \sqrt{16 \kappa\, d_{1r}}\, \enspace.
}
\end{align*}
Using the Cut-and-Paste property of private coin communication protocols
(Proposition~\ref{fact:cutandpaste}), we conclude that simultaneously
changing Bob's input from~$x[1,j]$ to~$x[1,l]$ and flipping the~$l$th
bit of~$x$ perturbs the transcript by no more than the individual changes.
\begin{align} \label{eq:2}
\expct_{(x[1,l],j,l) \leftarrow (X[1,L],J,L)} \; 
&  \fh{ M_r(U_l \,;\, v_j)}{ M_r(U'_l \,;\, v_l)} \nonumber \\
& = \quad \expct_{(x[1,l],j,l) \leftarrow (X[1,L],J,L)} \; 
          \fh{ M_r(U_l \,;\, v_l)}{M_r(U'_l \,;\, v_j)} \nonumber \\
& \leq \quad \sqrt{8 \kappa\, c_r} + \sqrt{16 \kappa\, d_{1r}}
             \enspace.
\suppress{
& \expct_{(x[1,l],j,l) \leftarrow (X[1,L],J,L)} \; 
  \fh{ M_r(x[1,l]\, X[l+1,n] \,;\, x[1,j])}{
  M_r(x[1,l-1]\, \bar{x}_l\, X[l+1,n] \,;\, x[1,l])} \nonumber \\
& = \quad \expct_{(x[1,l],j,l) \leftarrow (X[1,L],J,L)} \; 
          \fh{ M_r(x[1,l]\, X[l+1,n] \,;\, x[1,l])}{
          M_r(x[1,l-1]\,\bar{x}_l\, X[l+1,n] \,;\, x[1,j])} \nonumber \\
& \leq \quad \sqrt{8 \kappa\, c_r} + \sqrt{16 \kappa\, d_{1r}}
             \enspace.
}
\end{align}
Combining Eq.~(\ref{eq:1}) and Eq.~(\ref{eq:2}), and using the Triangle
Inequality we get
\begin{align*}
\expct_{(x[1,l],l) \leftarrow (X[1,L],L)} 
& \; \fh{ M_r(U_l \,;\, v_l)}{  M_r(U'_l \,;\, v_l)} 
\quad \leq \quad 4\sqrt{2 \kappa\, c_r}  +  4 \sqrt{\kappa\, d_{1r}}
                  \enspace.
\suppress{
& \expct_{(x[1,l],l) \leftarrow (X[1,L],L)} 
\; \fh{ M_r(x[1,l] X[l+1,n] \,;\, x[1,l])}
       {  M_r(x[1,l-1]\,\bar{x}_l\, X[l+1,n] \,;\, x[1,l] )} \nonumber \\
& \leq \quad 4\sqrt{2 \kappa\, c_r}  +  4 \sqrt{\kappa\, d_{1r}}
                  \enspace.
}
\end{align*}
Using Proposition~\ref{fact:ellhell}, we translate this back to a bound
on~$\ell_1$ distance:
\begin{align*} 
\left\| M_r(X \,;\, X[1,L]) \, X[1,L] \right.
&  \mbox{} - \left. M_r( X^{(L)} \,;\, X[1,L]) \, X[1,L] \right\| \\ 
& \leq \quad \expct_{(x[1,l],l) \leftarrow (X[1,L],L)} 
          \norm{ M_r(U_l \,;\, v_l)
                 -  M_r(U'_l \,;\, v_l)} \\
    & \leq \quad 16 \sqrt{\kappa\, c_r} + 8 \sqrt{2 \kappa\, d_{1r}}
                  \enspace.
\suppress{
& \norm{ M_r(X \,;\, X[1,L]) \, X[1,L]
         - M_r( X^{(L)} \,;\, X[1,L]) \, X[1,L] } \\ 
& \leq \quad \expct_{(x[1,l],l) \leftarrow (X[1,L],L)} 
          \norm{ M_r(x[1,l]\, X[l+1,n] \,;\, x[1,l])
                 -  M_r(x[1,l-1]\,\bar{x}_l\, X[l+1,n] \,;\, x[1,l] )} \\
    & \leq \quad 16 \sqrt{\kappa\, c_r} + 8 \sqrt{2 \kappa\, d_{1r}}
                  \enspace.
}
\end{align*}

Lemma~\ref{thm-close} follows by combining this with Eq.~(\ref{eqn-bound}).
\end{proofof}

We return to the lemmata whose proofs we had deferred.

\begin{lemma}
\label{thm-pairs}
$
\expct_{(x,j,l) \leftarrow (X,J,L)}
\; \fh{M_r(x \,;\, x[1,j])}{ M_r(x \,;\, x[1,l])}^2 
\quad \leq \quad 8 \kappa \, c_r.
$
\end{lemma}
\begin{proof}
Let us define a new random variable~$\tM_r$ jointly distributed with~$X$,
and independent of all other random variables, such that the joint
distribution of $X\tM_r$ is identical to the joint distribution of $XM_r^0$.
In particular, we have~$\tM_r(x) = \expct_{k \leftarrow K}
M_r (x \,;\, x[1,k])$.

By the Average Encoding Theorem, Proposition~\ref{thm-avg}, we have that
for every~$x \in \set{0,1}^n$,
\begin{eqnarray*}
\expct_{k \leftarrow K} \; \fh{M_r(x \,;\, x[1,k])}{\tM_r(x)}^2
    & \leq & \kappa \, \rI(K:M_r^0 \,|\, X = x) \enspace,
\end{eqnarray*}
where~$\kappa = \frac{\ln 2}{2}$. Averaging over~$x \leftarrow X$,
\begin{eqnarray*}
\expct_{(x,k) \leftarrow (X,K)} 
\; \fh{M_r(x \,;\, x[1,k])}{ \tM_r(x)}^2
    & \leq &  \kappa  \, \rI(K:M_r^0 \,|\, X)  \quad = \quad \kappa \, c_r \enspace.
\end{eqnarray*}
An immediate consequence is that
\begin{eqnarray*}
\expct_{(x,j) \leftarrow (X,J)} \; \fh{M_r(x \,;\, x[1,j])}{  \tM_r(x)}^2
    & \leq & 2 \, \kappa \, c_r \enspace, \qquad \text{and} \\
\expct_{(x,l) \leftarrow (X,L)} \; \fh{M_r(x \,;\, x[1,l]) }{ \tM_r(x)}^2
    & \leq & 2 \, \kappa \, c_r \enspace.
\end{eqnarray*}
By the Triangle Inequality, for any~$j \in [n/2]$, $l \in [n]-[n/2]$,
and~$x \in \set{0,1}^n$,
\begin{eqnarray*}
\lefteqn{ \fh{M_r(x \,;\, x[1,j])}{  M_r(x \,;\, x[1,l])}^2} \\
& \leq &  \left(\fh{M_r(x \,;\, x[1,j])}{  \tM_r(x)}
              + \fh{M_r(x \,;\, x[1,l])} { \tM_r(x)} \right)^2 \\
& \leq &  2 \, \fh{M_r(x \,;\, x[1,j])}{  \tM_r(x)}^2
          + 2 \, \fh{M_r(x \,;\, x[1,l])}{\tM_r(x)}^2 \enspace .
\end{eqnarray*}
Taking expectation over~$X,J,L$, we get the claimed bound.
\end{proof}

\begin{lemma}
\label{thm-flip}
$\expct_{(x[1,l],j,l) \leftarrow (X[1,L],J,L)} 
\; \fh{ M_r(U_l \,;\, v_j)}{ M_r(U'_l \,;\, v_j)}^2 
\quad \leq \quad 16 \kappa\, d_{1r} \enspace.$
\suppress{
We have
\begin{align*}
\expct_{(x[1,l],j,l) \leftarrow (X[1,L],J,L)} &
\; \fh{ M_r(x[1,l]\, X[l+1,n] \,;\, x[1,j])} 
{ M_r(x[1,l-1]\, \bar{x}_l\, X[l+1,n] \,;\, x[1,j])}^2 \\
    & \leq \quad 16 \kappa\, d_{1r} \enspace.
\end{align*}
}
\end{lemma}
\begin{proof}
This intuition behind this lemma is the same as that behind the
impossibility of ``random access
encoding''~\cite{Nayak99,AmbainisNTV02}, as we explain next.
Suppose we view the transcript as an
encoding of the bits of~$X$ not known to Bob, of which there are at
least~$n/2$. Since they are uniformly random, the net information in the
encoding about the bits is no more than the sum of the information
about the individual bits, even conditioned on the prefix. This follows by the 
superadditivity of mutual information for independent random variables
(equivalently, the Chain Rule,
Proposition~\ref{fact-chain}). This implies that, on average, the
encoding is very weakly correlated with the bits. The Average Encoding
Theorem (Proposition~\ref{thm-avg}) then implies that the
messages for two prefixes that differ in one bit are close to
each other, on average. We formalize this below.

We have
\begin{eqnarray}
\label{eqn-J}
d_{1r} n & \geq &  \rI(X : M_r^0 \,|\, X[1,K]) \nonumber \\
& = & \frac{1}{2} \; \expct_{j \leftarrow J} \,  \rI(X : M_r(X\,;\, X[1,J]) \,|\, X[1,J])  +  \frac{1}{2} \; \expct_{l \leftarrow L} \,  \rI(X : M_r(X\,;\, X[1,L]) \,|\, X[1,L]) \nonumber \\  
& \geq & \frac{1}{2} \; \expct_{j \leftarrow J} \,  \rI(X : M_r(X\,;\, X[1,J]) \,|\, X[1,J])   \enspace.  
\end{eqnarray}
Fix a sample point~$(x[1,j],j)$, with~$j \in [n/2]$.  By the Chain Rule
(Proposition~\ref{fact-chain}),
\begin{eqnarray}
\lefteqn{ \rI( X[j+1,n] : M_r(U_j \,;\, v_j) ) } \\
\nonumber
    & = & \sum_{l = j+1}^n 
          \rI( X_l : M_r(U_j \,;\, v_j) \,|\, X[j+1,l-1]) \\
\label{eqn-L}
    & \geq & \sum_{l = n/2 + 1}^n 
          \rI( X_l : M_r(U_j \,;\, v_j) \,|\, X[j+1,l-1]) 
          \enspace.
\suppress{
\nonumber
\lefteqn{ \rI( X[j+1,n] : M_r(x[1,j]\, X[j+1,n] \,;\, x[1,j]) ) } \\
\nonumber
    & = & \sum_{l = j+1}^n 
          \rI( X_l : M_r(x[1,j]\, X[j+1,n] \,;\, x[1,j]) \,|\, X[j+1,l-1]) \\
\label{eqn-L}
    & \geq & \sum_{l = n/2 + 1}^n 
          \rI( X_l : M_r(x[1,j]\, X[j+1,n] \,;\, x[1,j]) \,|\, X[j+1,l-1]) 
          \enspace.
}
\end{eqnarray}
Moreover, by the Triangle Inequality and the Average Encoding 
Theorem (Proposition~\ref{thm-avg}), for any given $x[1,l]$, 
with~$l \in [n] - [n/2]$,
\begin{align}
\label{eqn-flip}
\nonumber
& \fh{  M_r(U_l \,;\, v_j)}{M_r(U'_l \,;\, v_j) }^2 \\
\nonumber
& \leq \quad \bigl[ \; \fh{  M_r(U_l \,;\, v_j)}{M_r(U_{l-1} \,;\, v_j) } 
       + \fh{  M_r(U'_l \,;\, v_j)}{M_r(U_{l-1} \,;\, v_j) } \; \bigr]^2 \\
\nonumber
& \leq \quad 2 \; \bigl[ \; \fh{  M_r(U_l \,;\, v_j)}
        {M_r(U_{l-1} \,;\, v_j) }^2
        + \fh{  M_r(U'_l \,;\, v_j)}{M_r(U_{l-1} \,;\, v_j) }^2 \; \bigr] \\
& \leq \quad 4 \kappa \; \rI( X_l : M_r(U_{l-1} \,;\, v_j))
             \enspace .
\suppress{
\nonumber
& \fh{  M_r(x[1,l-1]\, x_l\, X[l+1,n] \,;\, x[1,j])}
     {M_r(x[1,l-1]\,\bar{x}_l\, X[l+1,n] \,;\, x[1,j]) }^2 \\
\nonumber
& \leq \quad \bigl[ \; \fh{  M_r(x[1,l-1]\, x_l\, X[l+1,n] \,;\, x[1,j])}
     {M_r(x[1,l-1]\, X_l\, X[l+1,n] \,;\, x[1,j]) }  \\
\nonumber
& \qquad  + \fh{  M_r(x[1,l-1]\, \bar{x}_l\, X[l+1,n] \,;\, x[1,j])}
     {M_r(x[1,l-1]\, X_l\, X[l+1,n] \,;\, x[1,j]) } \; \bigr]^2 \\
\nonumber
& \leq \quad 2 \; \bigl[ \; \fh{  M_r(x[1,l-1]\, x_l\, X[l+1,n] \,;\, x[1,j])}
     {M_r(x[1,l-1]\, X_l\, X[l+1,n] \,;\, x[1,j]) }^2 \\
\nonumber
& \qquad  + \fh{  M_r(x[1,l-1]\, \bar{x}_l\, X[l+1,n] \,;\, x[1,j])}
     {M_r(x[1,l-1]\, X_l\, X[l+1,n] \,;\, x[1,j]) }^2 \; \bigr] \\
& \leq \quad 4 \kappa \; \rI( X_l : M_r(x[1,l-1]\, X_l\, X[l+1,n] \,;\, x[1,j]))
             \enspace .
}
\end{align}
Combining Eqs.~(\ref{eqn-J}), (\ref{eqn-L}), and~(\ref{eqn-flip}), we
get
\begin{align*}
& \expct_{(x[1,l],j,l) \leftarrow (X[1,L],J,L)}
\; \fh{ M_r(U_l \,;\, v_j)}{ M_r(U'_l \,;\, v_j)}^2 \\
& \leq \quad  4 \kappa \; \expct_{(x[1,l-1],j,l) \leftarrow (X[1,L-1],J,L)} \;
              \rI( X_l : M_r(U_{l-1} \,;\, v_j)) \\
& = \quad  4 \kappa \; \expct_{(x[1,j],j,l) \leftarrow (X[1,J],J,L)} \; 
           \rI( X_l : M_r(U_j \,;\, v_j) \, | \, X[j+1,l-1]) \\
& \leq \quad \frac{8\kappa}{n} \; \rI( X : M_r(X \,;\, X[1,J]) \, | \, X[1,J]) 
\quad \leq \quad 16 \kappa \; d_{1r} \enspace,
\suppress{
& \expct_{(x[1,l],j,l) \leftarrow (X[1,L],J,L)}
\; \fh{ M_r(x[1,l]\, X[l+1,n] \,;\, x[1,j])} 
{ M_r(x[1,l-1]\,\bar{x}_l\, X[l+1,n] \,;\, x[1,j])}^2 
    \\
& \leq \quad  4 \kappa \; \expct_{(x[1,l-1],j,l) \leftarrow (X[1,L-1],J,L)} \, 
              \rI( X_l : M_r(x[1,l-1]\, X_l\, X[l+1,n] \,;\, x[1,j]))
    \\
& = \quad  4 \kappa \; \expct_{(x[1,j],j,l) \leftarrow (X[1,J],J,L)} \, 
           \rI( X_l : M_r(x[1,j]\, X[j+1,n] \,;\, x[1,j]) \, | \, X[j+1,l-1])
    \\
& \leq \quad \frac{8\kappa}{n} \; \rI( X : M_r(X \,;\, X[1,J]) \, | \, X[1,J]) 
\quad \leq \quad 16 \kappa \; d_{1r} \enspace,
}
\end{align*}
as claimed.
\end{proof}

\section{The connection with streaming algorithms}
\label{sec-streaming}

Streaming algorithms are algorithms of a simple form, intended to
process massive problem instances rapidly, ideally using space that is
of smaller order than the size of the input.  A {\em pass\/}
on an input $x\in\Sigma^n$, where~$\Sigma$ is some alphabet,
means that $x$ is read as an {\em input
  stream\/}~$x_1,x_2,\ldots,x_n$, which arrives sequentially, i.e.,
letter by letter in this order.
\begin{definition}[Streaming algorithm]
Fix an alphabet $\Sigma$.  A (unidirectional) $T$-pass {\em streaming 
algorithm\/}~\textbf{A} with space $s(n)$ and time $t(n)$ is an algorithm 
such that for every input stream~$x\in\Sigma^n$:
\begin{enumerate}
\item  
\textbf{A}  performs~$T$ sequential passes on~$x$ in the order~$x_1,
x_2, \dotsc, x_n$,
\item 
\textbf{A}  maintains a memory space of size~$s(n)$ bits
while reading $x$,
\item 
\textbf{A}  has running time at most~$t(n)$ per letter~$x_i$, and
\item 
\textbf{A}  has pre-processing and post-processing time at most~$t(n)$.
\end{enumerate}
We say that \textbf{A} is {\em bidirectional\/} if it is allowed to
read the input in the reverse order, after reaching the last letter.
Then the parameter~$T$ is the total number of passes in
either direction.
\end{definition}
In general, the pre- and post-processing times of a streaming algorithm
may be different, and may differ from the running time per letter. Since
the results in this section apply to streaming algorithms regardless of
their time complexity, we choose not to make this finer distinction.

We refer the reader to the text~\cite{Muthukrishnan05} for a more
thorough introduction to streaming algorithms.

Recall that in a two-party communication protocol for \AIndex, one
party, Alice, has
an~$n$-bit string~$x$, and the other party, Bob, has an integer~$k \in
[n]$, the prefix~$x[1,k-1]$ of~$x$, and a bit~$b \in \set{0,1}$. Their
goal is to compute the function~$f_n(x,(k, x[1,k-1], b)) = x_k \xor
b$, i.e., to determine whether~$b = x_k$ or not, by engaging in a
two-party communication protocol.

The relationship between streaming algorithms for $\dyck(2)$ and
communication protocols for~$f_n$ is captured by a reduction due to 
Magniez, Mathieu,
and Nayak~\cite{MagniezMN10}. The reduction was originally described 
only for one-pass streaming algorithms, but extends readily to
unidirectional multi-pass algorithms. For completeness, we include
a proof of this theorem here.
\begin{theorem}
\label{thm-mmn}
Suppose there is a randomized unidirectional streaming algorithm for $\dyck(2)$ 
with~$T$ passes that uses space~$s$ for instances of length at most~$4n^2$,
and has worst-case two-sided error~$\delta$. Then there is a two-party
communication protocol~$\Pi$ for the \AIndex\ function~$f_n$ that makes
error at most~$\delta$ on the uniform distribution~$\mu$ over its
inputs, and has information costs~$\ic^\sA_{\mu_0}(\Pi) \le sT$ for
Alice and~$\ic^\sB_{\mu_0}(\Pi) \le sT/n$ for Bob, with respect to the
uniform distribution~$\mu_0$ over~$f_n^{-1}(0)$.
\end{theorem}
\begin{proof}
For any string $z = z_1 \dotsb z_n \in\set{a,b}^n$, let $\overline{z}$ 
denote the matching string~$\overline{z_n} \, \overline{z_{n-1}} \,
\dotsb \overline{z_1}$
corresponding to $z$. Let~$z[i,j]$ denote the substring~$z_i z_{i+1}
\dotsb z_j$ if~$1 \leq i \leq j \leq n$, and the empty string~$\epsilon$
otherwise. We abbreviate~$z[i,i]$ as~$z[i]$ if~$1 \le i \le n$.

We focus on a subset of instances for $\dyck(2)$ defined as follows.
Let~$n$ be a positive integer. Consider strings of the form
\begin{eqnarray}
\label{eqn-instance}
w &=&   x^1 \, \overline{y^1} \, \overline{z^1} \, z^1 \, y^1 \ 
        x^2 \, \overline{y^2} \, \overline{z^2} \, z^2 \, y^2 \ 
        \dotsb \ x^n \, \overline{y^n} \, \overline{z^n} \, z^n \, y^n \ 
        \overline{x^n}  \ \dotsb\ \overline{x^{2}} \ \overline{x^1}
        \enspace,
\end{eqnarray}
where for every $i$, $x^i\in \{ 0,1\}^n$, $y^i={x^i[n-k^i+2,n]}$ for some
$k^i\in\{1,2,\ldots,n\}$, and $z^i\in \{ a,b\}$. The string $w$ is in
$\dyck(2)$ if and only if, for every $i$, $z^i=x^i[n-k^i+1]$. Note that
these instances have length in the interval~$[2n(n+1), 4n^2]$.
Figure~\ref{fig-instance} depicts an instance of this form.
\begin{figure}[ht]
\centerline{\includegraphics[width=300pt]{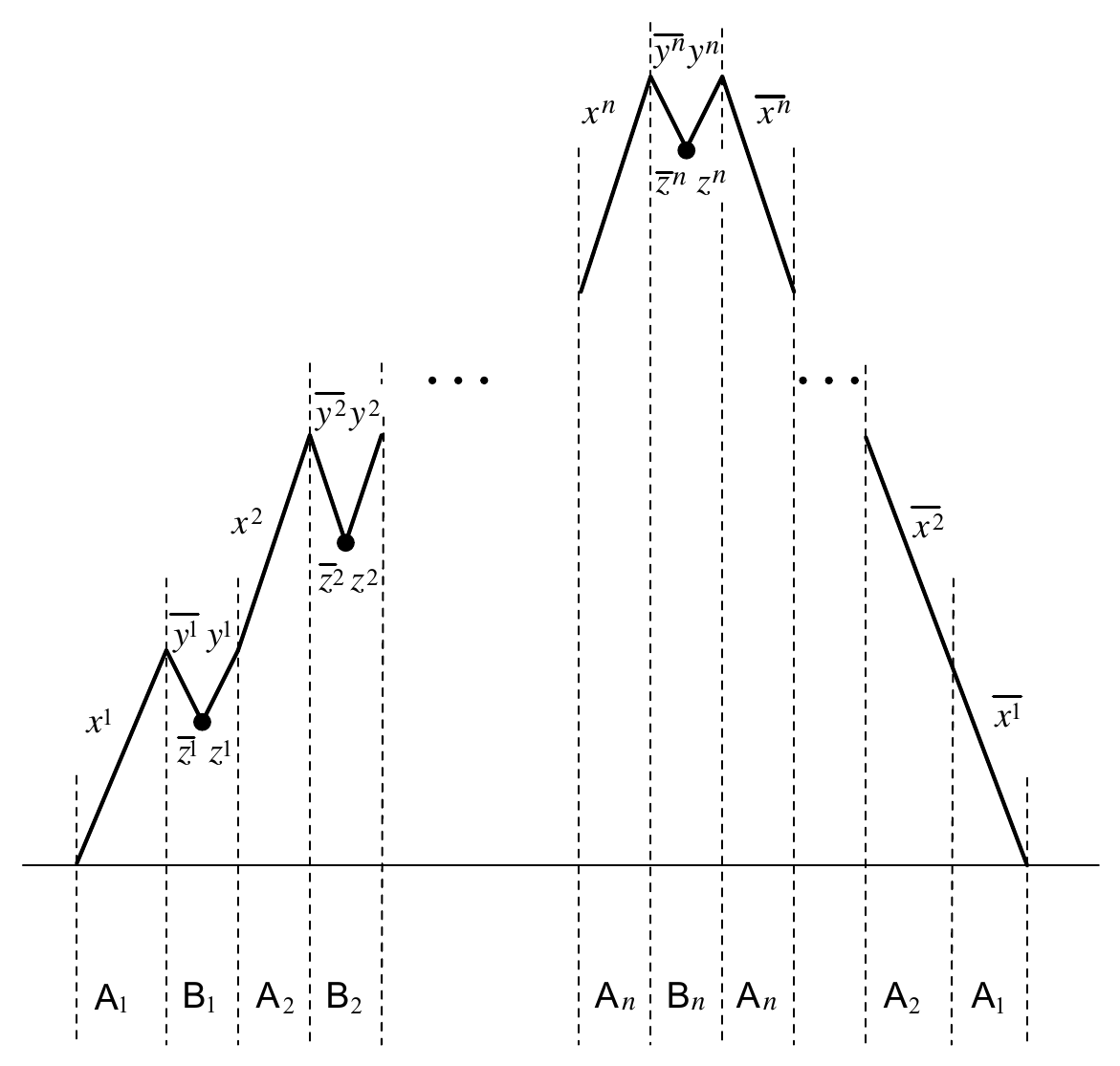}}
\caption{ \label{fig-instance}
An instance of the form described in Eq.~(\ref{eqn-instance}). A line
segment with positive slope denotes a string over~$\set{a,b}$, and a
segment with negative slope denotes a string over~$\set{\overline{a},
\overline{b}}$. A solid dot depicts a pair of the form~$\overline{z} z$
for some~$z \in \set{a,b}$. The entire string is distributed amongst~$2n$
players~$\sA_1, \sB_1, \sA_2, \sB_2,  \dotsc, \sA_n, \sB_n$ in
a communication protocol for~$\Asc(n)$ as shown.  }
\end{figure}

Intuitively, recognizing strings of the form~$w$ is difficult in
one pass with space~$\order(n)$.
After reading $x^i$, the streaming algorithm does not have enough space to 
store this string so as to be able to check the bit at unknown 
index~$(n-k_i+1)$. Moreover, after reading $\overline{y_n}$ it does not 
have enough space to store information about all indices~$k^1,k^2,\dotsc,
k^n$. When it reads $\overline{x^n} \dotsb \overline{x^2} \, \overline{x^1}$
it therefore misses out on its second chance to check whether~$z^i 
=x^i[n-k^i+1]$ for every $i$. When the algorithm is allowed a larger
number of passes~$T$ in the same direction, it may adopt a more
sophisticated strategy. Nevertheless, the same intuition carries over 
with a tighter bound of~$\order(n/T)$ on the space.

We observe that a space~$s$ streaming algorithm gives rise to a multiparty
communication protocol for the problem $\Asc(n)$, which is the logical OR
of~$n$ independent instances of the \AIndex\ function~$f_n$.
In more detail, in the problem~$\Asc(n)$ there are~$2n$ players
$\sA_1,\sA_2,\dotsc,\sA_n$ and
$\sB_1,\sB_2,\ldots,\sB_n$. Player~$\sA_i$ is given~$x^i\in\{0,1\}^n$,
player~$\sB_i$ is given $k^i\in [n]$, a bit~$z^i$, and the
prefix~$x^i[1, k^i-1]$ of $x^i$. Let $\mathbf{x}=(x^1,x^2,\ldots,x^n)$,
$\mathbf{k}=(k^1,k^2,\ldots,k^n)$, and $\mathbf{z}=(z^1,z^2,\ldots,z^n)$.

The goal of the communication protocol is to compute
\[
F_n(\mathbf{x},\mathbf{k},\mathbf{z})
    \quad = \quad \bigvee_{i=1}^n f_n(x^i,k^i,z^i)
    \quad = \quad \bigvee_{i=1}^n (x^i[k^i] \xor z^i) \enspace,
\]
which is $0$ if $x^i[k^i]=z^i$ for all $i$, and $1$ otherwise.
The communication between the~$2n$ parties is required to be~$T$ sequential 
iterations of communication in the following order, for some~$T \ge 1$:
\begin{equation}
\label{eqn-commn}
\sA_1 \rightarrow \sB_1 \rightarrow
\sA_2 \rightarrow \sB_2 \rightarrow
\dotsb
\sA_n \rightarrow \sB_n \rightarrow
\sA_n \rightarrow \sA_{n-1} \rightarrow
\dotsb \rightarrow \sA_2 \rightarrow \sA_1 \enspace.
\end{equation}
In other words, for~$t = 1, 2, \dotsc, T$,
\begin{itemize}
\item[\ --] for $i$ from $1$ to $n-1$, player~$\sA_i$ sends message
  $M_{\sA_i, t}$ to $\sB_i$, then $\sB_i$ sends message~$M_{\sB_i,t}$ to
$\sA_{i+1}$,
\item[\ --] $\sA_n$ sends message $M_{\sA_n, t}$ to $B_n$,
\item[\ --] $B_n$ sends message $M_{B_n,t}$ to $\sA_{n}$,
\item[\ --] for $i$ from $n$ down to 2, $\sA_i$ sends message
  $M'_{\sA_i,t}$ to $\sA_{i-1}$.
\end{itemize}
At the end of the~$T$ iterations, $\sA_1$ computes the output.

There is a one-to-one correspondence between inputs to $\dyck(2)$ of the
form in  Eq.~(\ref{eqn-instance}) and the inputs to~$\Asc(n)$.
This arises from a partition of the word among~$2n$ players as depicted in 
Figure~\ref{fig-instance}. For ease of notation, 
the strings~$x^i$ in~$\Asc(n)$ are taken to be the ones in $\dyck(2)$
with the bits \emph{in reverse order\/}. This switches the suffixes~$y^i$ 
with prefixes of the same length. 

The following is immediate.
\begin{lemma}
A unidirectional $T$-pass streaming algorithm for~$\dyck(2)$ with 
space~$s$ implies a communication protocol for~$\Asc(n)$ with~$T$
iterations of communication as above, in which every message is of
length~$s$. Moreover, on any input, the probability of error of the 
protocol is the same as that of the algorithm.
\end{lemma}
\begin{proof}
In each of the~$T$ iterations, a player simulates the streaming 
algorithm on his/her part of the input, and sends the length~$s$
workspace to the next player in the sequence. The final
player $\sA_1$ gives the output of the algorithm as that of the
protocol.
\end{proof}

We prove a {\em direct sum\/} result that captures the relationship of 
$\Asc(n)$ to solving~$n$ instances of the more ``primitive''
problem \AIndex. 
\suppress{
Recall that in the $\AIndex$ problem~$f_n$, Alice
has an $n$-bit string $x\in\{0,1\}^n$, and Bob has an integer
$k\in [n]$, a bit $c$ and the prefix $x[1, k-1]$ of $x$.  The goal is
to compute the Boolean function~$f_n(x,k,c)= (x[k] \xor c)$ which is~$0$
if $x[k]=c$, and $1$ otherwise. 
}
The direct sum result is proven using the superadditivity of mutual
information for inputs~$(x^i,k^i,z^i)$ picked independently from the
uniform distribution~$\mu_0$ over~$f_n^{-1}(0)$. The use of this 
``easy'' distribution collapses the function~$\Asc(n)$ to an instance 
of \AIndex\ in any chosen coordinate. The direct sum result allows us to
choose a coordinate with small information cost, which proves the
theorem.

Consider an instance~$(\mathbf{X}, \mathbf{K}, \mathbf{Z})$ of \Asc$(n)$
distributed according to~$\mu_0^n$
over~$(\{0,1\}^n\times[n]\times\{0,1\})^n$, where
$\mathbf{X}=(X^1,X^2,\ldots,X^n)$, $\mathbf{K}=(K^1,K^2,\ldots,K^n)$
and $\mathbf{Z}=(Z^1,Z^2,\ldots,Z^n)$.

Let $\tPi$ be a public-coin randomized protocol for $\Asc(n)$ derived
from a unidirectional~$T$-pass streaming algorithm for $\dyck(2)$.
Assume it has worst-case error~$\delta$, and that each message 
is of length at most~$s$. For each~$j \in [n]$, we construct a
protocol~$\Pi_{j}$ as follows for the \AIndex\ function~$f_n$.
Let~$(x,k,c)$ be the input for~\AIndex.
\begin{enumerate}
\item Alice sets $\sA_j$'s input~$x^j$ to her input~$x$.
\item Bob sets $\sB_j$'s input~$(k^j,x^j[1,k^j-1],z^j)$ to his
  input~$(k,x[1,k-1],c)$.
\item Alice and Bob generate, using public coins, $X^i$ uniformly at
random from~$\set{0,1}^n$, independently for all~$i > j$, and~$(X^i,K^i,Z^i)$ distributed according to~$\mu_0$, independently for
all~$i < j$.
\item Bob generates~$K^i$ uniformly and independently for~$i > j$, 
  using private coins.  Then Bob sets $Z^i=X^i[k^i]$ for $i>j$, so
  that $(X^i,K^i,Z^i)$ are distributed according to $\mu_0$, independently
for all $i > j$.
\item Alice and Bob simulate the protocol~$\tPi$ by executing the roles
of players~$(\sA_i,\sB_i)_{i = 1}^n$ as follows. In the~$t$th
iteration of communication in the order described in
Eq.~(\ref{eqn-commn}),
\begin{enumerate}
\item Alice runs $\tPi$ until she generates the message~$M_{\sA_j,t}$ from
  player~$\sA_j$. She sends this message to Bob.
\item Bob continues running $\tPi$ until he generates the message
 $M_{\sB_n,t}$ from player~$\sB_n$.  He sends this message to Alice.
\item Alice completes the rest of the~$t$th iteration of $\tPi$ until she
generates the message~$M'_{\sA_2,t}$ from player~$\sA_2$, and moves to
the next iteration of~$\tPi$ (if any).
\end{enumerate}
At the end of the~$T$th iteration, Alice completes the rest of the
protocol~$\tPi$ and produces as output for~$\Pi_j$, the output of 
player $\sA_1$ in~$\tPi$.
\end{enumerate}
By definition of the distribution~$\mu_0$, we have~$f_n(X^i,K^i,Z^i) =
0$ for all~$i \neq j$.  So $F_n(\mathbf{X},\mathbf{K},\mathbf{Z}) =
f_n(x,k,c)$, and each protocol $\Pi_j$ computes the function~$f_n$, i.e.,
solves \AIndex, with worst-case error at most~$\delta$.

Note that in the simulation of~$\tPi$ by Alice and Bob above, the random
variables~$(X^i,K^i,Z^i)$ for~$i < j$ are used only by Alice, and could 
have been generated by Alice using private coins. Making these random
variables public does not affect the correctness of~$\Pi_j$, but turns
out to be convenient in deriving the direct sum result.

Let~$R$ denote the public coins used in the protocol~$\tPi$.
Let~$\mathbf{M}$ denote the sequence of~$T$ random
variables~$M_{\sB_n,1} M_{\sB_n,2} \dotsb M_{\sB_n,T}$, viz., the
messages sent by~$\sB_n$ over all the iterations.  By the Chain Rule 
(Proposition~\ref{fact-chain}),
\begin{eqnarray*}
\rI(\mathbf{K}\mathbf{Z}: \mathbf{M} \;|\;  \mathbf{X} R)
    & = & \sum_{j = 1}^n \rI({K}^{j} Z^j : \mathbf{M} \;|\; 
          \mathbf{X} R \, {K}^{1}Z^1 \dotsb {K}^{j-1} Z^{j-1}) \enspace.
\end{eqnarray*}

Let $R_j=(R,(X^i)_{j\neq i},(K^i,Z^i)_{i < j})$. These are all the
public random coins used in the protocol~$\Pi_j$, and any further
random coins are used only by Bob privately to generate~$(K^i,Z^i)_{i > j}$.
In particular, Alice does not use any private coins and her messages are 
(deterministic) functions of~$X^j R_j$ and the messages received from
Bob. Thus, for all~$j$
\begin{eqnarray*}
\ic^{\sB}_{\mu_0}(\Pi_{j}) 
    & = & \rI({K}^{j} Z^j : \mathbf{M} \;|\;  X^j R_j) \\
    & = & \rI({K}^{j} Z^j : \mathbf{M} \;|\;  \mathbf{X} R \, {K}^{1}Z^1
              \dotsb {K}^{j-1} Z^{j-1}) \enspace,
\end{eqnarray*}
and we have the direct sum result
\begin{eqnarray*}
\sum_{j = 1}^n \ic^{\sB}_{\mu_0}(\Pi_{j})
    & = & \rI(\mathbf{K}\mathbf{Z} : \mathbf{M} \;|\;  \mathbf{X} R)
          \enspace.
\end{eqnarray*}
Furthermore, $\mathbf{M}$ has length at most~$sT$, so that
\begin{eqnarray*}
\sum_{j = 1}^n \ic^{\sB}_{\mu_0}(\Pi_{j})
    & \le & sT \enspace,
\end{eqnarray*}
and there is a~$j_0 \in [n]$ such that~$\ic^{\sB}_{\mu_0}(\Pi_{j_0})
\le sT/n$. We also have, by the Chain Rule (Proposition~\ref{fact-chain}),
\begin{eqnarray}
\nonumber
\ic^{\sA}_{\mu_0}(\Pi_{j_0})
    & = & \rI( X^{j_0} : M_{\sA_{j_0},1} M_{\sA_{j_0},2} \dotsb
          M_{\sA_{j_0},T} \, \mathbf{M} \;|\;  
          K^{j_0} Z^{j_0} \, R_{j_0}) \\ 
\nonumber
    & = & \sum_{t = 1}^T \left[ \rI( X^{j_0} : M_{\sA_{j_0},t} \;|\;
          K^{j_0} Z^{j_0} \, R_{j_0} M_{\sA_{j_0},1} M_{\sB_n,1}
          \dotsb M_{\sA_{j_0},t-1} M_{\sB_n,t-1} ) \right. \\
\nonumber
    &   & \left. \mbox{} + \rI( X^{j_0} : M_{\sB_n,t} \;|\; 
          K^{j_0} Z^{j_0} \, R_{j_0} M_{\sA_{j_0},1} M_{\sB_n,1}
          \dotsb M_{\sA_{j_0},t-1} M_{\sB_n,t-1} M_{\sA_{j_0},t}) \right] \\
\label{eqn-ica}
    & = & \sum_{t = 1}^T \rI( X^{j_0} : M_{\sA_{j_0},t} \;|\;
          K^{j_0} Z^{j_0} \, R_{j_0} M_{\sA_{j_0},1} M_{\sB_n,1}
          \dotsb M_{\sA_{j_0},t-1} M_{\sB_n,t-1} ) \enspace,
\end{eqnarray}
since Bob's~$t$th message is independent of Alice's input, conditioned on
his input, the public randomness, and the transcript until Alice's~$t$th
message. Since the length of each message~$M_{\sA_{j_0},t}$ is bounded
by~$s$, Eq~(\ref{eqn-ica}) implies
\begin{eqnarray*}
\ic^{\sA}_{\mu_0}(\Pi_{j_0}) & \le & sT \enspace.
\end{eqnarray*}
The protocol~$\Pi_{j_0}$ is the protocol claimed by the theorem.
\end{proof}

The information cost trade-off in Theorem~\ref{thm-main} implies that
any streaming algorithm that makes a ``small'' number of passes over the
input requires a ``large'' amount of space.
\begin{corollary}
\label{thm-dyck}
Any randomized unidirectional $T$-pass streaming algorithm 
for $\dyck(2)$ that has worst-case two-sided error~$\delta < 1/4$
uses space at least
\[
\frac{\floor{\sqrt{N}/2}}{T} \times \frac{1}{3 + 2 \sqrt{2} }
\left[ 
\frac{1-4\delta}{4\sqrt{\ln 2}}
                 - \left( \frac{\rH(2\delta)}{\floor{\sqrt{N}/2}}
                   \right)^{1/2}
\right]^2
\]
on instances of length~$N$.
\end{corollary}

\section{Quantum information cost of Augmented Index}
\label{sec-quantum}

We now turn to quantum communication. We present the necessary background
on quantum information theory in Section~\ref{sec-qinfo-theory}, and discuss
quantum protocols and information cost in Section~\ref{sec-qcommn}.
In Section~\ref{sec-qlb},
we show how the notion of average encoding may be applied also to quantum 
protocols for \AIndex. The analysis of quantum protocols for \AIndex\
involves a number of additional additional subtleties, which are also 
described along the way.

\subsection{Quantum information theory basics} 
\label{sec-qinfo-theory}

We continue the use of capital letters to denote random variables. We
see these as special cases of quantum states, which are trace one 
positive semi-definite matrices. Indeed, random variables may be viewed as 
quantum states that are diagonal in a canonical basis. Quantum states 
are also denoted by capital letters~$P,Q$, etc.

The trace distance~$\trnorm{A-B}$ between two quantum states~$A,B$
over the same Hilbert space is the metric induced by the trace
norm~$\trnorm{M} = \trace \sqrt{ M^\adjoint M }$. The fidelity between the
two states is defined as $\rF(A,B) =  \trnorm{ \sqrt{A} \sqrt{B} }$. 
The Bures distance~$\fh{A}{B}$ between the states is a metric arising 
from fidelity, and is defined as
\[
\fh{A}{B} \quad = \quad  \left[ 1 -  \rF(A,B)\right]^{1/2} \quad = \quad 
    \left[ 1 - \trnorm{ \sqrt{A} \sqrt{B} } \right]^{1/2} 
    \enspace.
\]
This metric generalizes Hellinger distance to quantum states; when~$A,B$
are random variables, Bures distance coincides with Hellinger distance.
For pure states $\ket{\psi_1}, \ket{\psi_2}$ we use
$\fh{\ket{\psi_1}}{\ket{\psi_2}}$ as shorthand for
$\fh{\ket{\psi_1}\bra{\psi_1}}{\ket{\psi_2}\bra{\psi_2}} $.  Bures
distance is related to trace distance in the following manner (see,
e.g., Lemma~II.6 in Ref.~\cite{KlauckNTZ07}):
\begin{proposition}
\label{lem-trbures} 
Let $P,Q$ be quantum states over the same Hilbert space. Then
\[
\fh{P}{Q}^2 
    \quad \leq \quad \frac{1}{2} \trnorm{P-Q} 
    \quad \leq \quad \sqrt{2} \; \fh{P}{Q} \enspace.
\]
\end{proposition}

In the following, let~$(p_x),(q_y)$ be distributions over the finite
sample spaces~$\cS, \cS'$, respectively.

The Bures distance satisfies the following property.
\begin{proposition} 
\label{lem-convexity} 
Let $P_x,Q_x$ be quantum states over the same finite Hilbert space for
each~$x \in \cS$. Let~$P = \sum_{x \in \cS} p_x \density{x} \tensor
P_x$, and~$Q = \sum_{x \in \cS} p_x \density{x} \tensor Q_x$.
Then
\[
\fh{P}{Q}^2 \quad = \quad \sum_{x \in \cS} p_x \, \fh{P_x}{Q_x}^2
\enspace.
\]
\end{proposition}
This may be verified readily by the definition of the Bures distance,
but may also be derived as an immediate consequence of the strong 
concavity property of fidelity~\cite[Theorem~9.7, p.~414]{NielsenC00}.

The Local Transition Theorem due to Uhlmann~\cite{NielsenC00} helps us
find purifications of quantum states that achieve the Bures distance
between them.
\begin{proposition}[Local Transition Theorem]
\label{fact:localtrans}
Let $\ket{\psi_1}$ and $\ket{\psi_2}$ be two pure states in a tensor
product~$\mathcal{H}_1 \otimes \mathcal{H}_2$ of Hilbert spaces. Then
there exists a unitary operator $U$ on $\mathcal{H}_1$ such that
\[
\fh{(U \otimes \identity_{\mathcal{H}_2}) \, \ket{\psi_1} } 
{\ket{\psi_2} } \quad = \quad 
    \fh{ \trace_{\mathcal{H}_1} \ket{\psi_1} \bra{\psi_1} }
       {\trace_{\mathcal{H}_1} \ket{\psi_2} \bra{\psi_2}} \enspace .
\]
\end{proposition}

We rely on a number of standard results from quantum information
theory in this work.  For a comprehensive introduction to the subject,
we refer the reader to a text such as~\cite{NielsenC00}.

Let~$\rS(P)$ denote the von Neumann entropy of the quantum state~$P$,
and~$\rI(P:Q)$ denote the mutual information between the two parts of
a joint quantum state~$PQ$.

For a joint quantum state~$XQ = \sum_{x \in \cS} p_x \density{x}
\tensor Q_x$ we define the conditional von Neumann entropy as~$\rS(Q
\,|\, X) = \sum_{x \in \cS} p_x \, \rS(Q_x)$. Similarly, for a joint
state~$XPQ = \sum_{x \in \cS} p_x \density{x} \tensor (P Q)_x$,
where~$(PQ)_x$ is a joint state for each~$x \in \cS$, we define the
conditional mutual information as
\[
\rI(P:Q \,|\, X) \quad = \quad 
    \rS(P \,|\, X) + \rS(Q \,|\, X) - \rS(PQ \,|\, X)
\enspace.
\]

The chain rule for mutual information states:
\begin{proposition}[Chain rule] 
\label{fact-chain-quantum}
Let $XYQ = \sum_{x \in \cS, y \in \cS'} p_x q_y \density{xy} \tensor
Q_{xy}$ be a joint quantum state. Then
$$ \rI(XY:Q)  \quad = \quad  \rI(X:Q) + \rI(Y : Q \,|\, X)  \enspace . $$
\end{proposition}
It follows directly from the identity~$\rS(XQ) = \rS(X) + \rS(Q|X)$ for
joint states~$XQ$ of the form~$XQ = \sum_{x \in \cS} p_x \density{x}
\tensor Q_x$.

The Average Encoding Theorem~\cite{KlauckNTZ07,JainRS03b} also holds
for quantum states. (In fact, it was first formulated in the context of
quantum communication.)
\begin{proposition}[Average encoding theorem]
\label{thm-avg-quantum} 
Let $XQ = \sum_{x \in \cS} p_x \density{x} \tensor Q_x$ be a joint
quantum state. Then,
\[
\expct_{x \leftarrow X} \; \fh{Q_x}{Q}^2 \quad \leq \quad \kappa \;
\rI(X:Q) \enspace,
\]
where~$\kappa$ is the constant~$\frac{\ln 2}{2}$.
\end{proposition}

\subsection{Quantum communication and information cost}
\label{sec-qcommn}

We briefly describe the model of two-party quantum communication,
\emph{{\`a} la\/} Yao~\cite{Yao93}. We only consider protocols with 
classical inputs and outputs. For the basic elements of quantum
computation, we refer the reader to a text such as~\cite{NielsenC00}.

Informally, two ``players'', Alice and Bob, hold some number of qubits.
When the protocol starts, Alice holds a classical input represented by 
a bit string~$x \in \cX$ and similarly Bob holds~$y \in \cY$.
The qubits in the workspace of the two parties are initialized to a 
state~$\ket{\Phi}$ that is independent of the inputs~$x,y$, and may be
entangled across the parties.
The protocol consists of some number~$t \ge 1$ of rounds of message
exchange, in which the two players ``play'' alternately. Any party
may be the first to play. Suppose it is Alice's turn to play. She applies
a unitary operator to her workspace qubits, which depends on her 
input~$x$ and the round. Then, Alice sends some of her workspace
qubits to Bob. In the next round, Bob's local computation thus
involves some qubits previously in Alice's control.
At the end of the~$t$ rounds of message exchange, the player to
receive the last message, say Bob, observes the qubits in his possession
according to a measurement that may depend on his input~$y$. The
measurement outcome is considered to be the output of the protocol.

More formally, a two-party quantum communication protocol~$\Pi$ is
specified as follows. The protocol uses some~$N$ qubits, for some
positive integer~$N$, so that the associated state space
is~$(\complex^2)^{\tensor N}$. We view this space as a tensor product
space~$\cA \tensor \cH_{\sA, i} \tensor \cH_{\sB,i} \tensor \cB$, for
each~$i = 0, 1, \dotsc, t$, with the initial factorization given by~$i =
0$, and the factorization at the end of the~$j$th round given ~$i = j$.
This factorization reflects the ownership of the qubits. The space~$\cA$
contains Alice's input, $\cB$ contains Bob's input, and the
spaces~$\cH_{\sA,i}$ and~$\cH_{\sB,i}$ correspond to Alice's and Bob's 
workspace qubits at the end of round~$i$, respectively.

The qubits in space~$\cA$ are initialized to~$\ket{x}$,
and those in~$\cB$ are initialized to~$\ket{y}$. The qubits in the
space~$\cH_{\sA, 0} \tensor\cH_{\sB,0}$ are initialized
to a possibly entangled state~$\ket{\Phi}$ that is independent of the
inputs.  The initial joint state is thus~$\ket{x} \tensor \ket{\Phi} 
\tensor \ket{y}$.

The protocol specifies the number~$t$ of messages sent, and the player
that sends the first message. Suppose it is Alice's turn to play in 
round~$i$, with~$i \ge 1$. The workspace of the two players just
before the round factors as~$\cH_{\sA,i-1} \tensor \cH_{\sB,i-1}$. Alice
applies a unitary operator~$V_{i,x}$ to the qubits in~$\cH_{\sA,i-1}$.
Note that her operator depends on her input~$x$ and the round.  (Later,
we imagine running the protocol on superpositions of inputs. In this 
case, we think of Alice as applying the unitary~$V_i = \sum_x
\density{x} \tensor V_{i,x}$ to the qubits in the space~$\cA \tensor
\cH_{\sA,i-1}$.) Then, Alice sends some of her qubits, corresponding to
the space~$\cM_i$, to Bob. That is, the space~$\cH_{\sA,i-1}$ factors 
as~$\cH_{\sA,i} \tensor \cM_i$, and~$\cH_{\sB,i} = \cM_i \tensor
\cH_{\sB,i-1}$. 

After the~$t$th message is sent, the recipient, say Bob, observes the 
qubits corresponding to~$\cH_{\sB,t}$ according to a POVM (positive operator valued measurement) that depends
on his input~$y$. The output of the protocol is the measurement outcome,
and we denote the corresponding random variable by~$\Pi(x,y)$.
Figure~\ref{fig-protocol} depicts such a two-party protocol.

\begin{figure}[ht]
\centerline{\includegraphics[width=250pt]{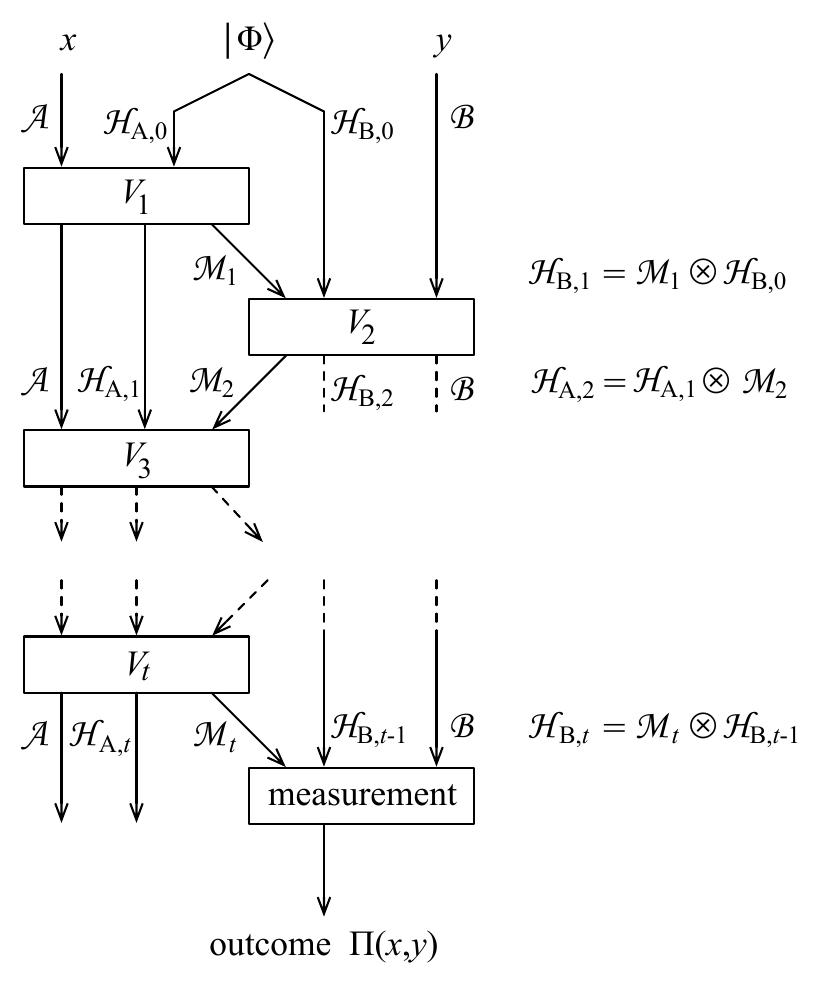}}
\caption{ \label{fig-protocol}
A quantum two-party communication protocol with~$t$ messages,
inputs~$x,y$ and shared initial state~$\ket{\Phi}$.}
\end{figure}

We emphasize that the input qubits in the protocol are \emph{read only\/},
and that there are no intermediate measurements. A more general protocol 
may be transformed into this form by appealing to standard techniques in
quantum computation~\cite{BernsteinV97}.

In this article, we are concerned with protocols designed to compute a
bi-variate Boolean function~$f : \cX \times \cY \rightarrow \set{0,1}$.
As for classical protocols, the probability of correctness 
(or \emph{success\/}) of a protocol on input~$x,y$ is~$\Pr[\Pi(x,y) =
f(x,y)]$. We consider inputs drawn
from a joint distribution~$XY$, in which case the success probability
is~$\Pr[\Pi(X,Y) = f(X,Y)]$. The probability of the complementary event
is called the \emph{error\/} of the protocol on the distribution~$XY$.

As in the classical case,  there is no
canonical measure of quantum information leaked by a protocol, and this 
notion is a topic of active research.  The choice of the measure is 
driven by a motivating application and the ease with which we can analyze
it. We typically strike a balance between these opposing forces.

A significant difference between classical and quantum information
costs arises because the no cloning principle~\cite[p.~532]{NielsenC00}
prevents the two parties from keeping a copy of the messages. A natural
notion of a transcript that encapsulates the history of a quantum
protocol is instead the sequence of the joint states after each
message exchange. Correspondingly, the notion of information cost is
also different from the one in the classical case.

Consider a quantum communication protocol~$\Pi$ with a total 
of~$t$ messages, beginning with Alice and alternating with Bob.
We emphasize that the input qubits in~$\Pi$ are read-only.
The first player is assumed to be Alice solely to eliminate
awkwardness in defining and referring to quantum information cost.
The assumption may be removed without affecting the results in this
article. Alternatively, if Bob starts, we may modify the
protocol so that Alice sends a single qubit in a fixed state,
say~$\ket{0}$, at the beginning. This does not affect the information
cost, but increases the number of messages by one.

Let~$\lambda$ be a probability distribution over~$\cX \times \cY$, 
and let random variables~$XY$ be distributed according to~$\lambda$.
Let~$P_i Q_i$ denote the joint state of Alice and Bob's workspace
\emph{immediately after\/} the~$i$th message is sent, in a
protocol~$\Pi$ when we start with the inputs~$XY$. In analogy with the
classical case, we may define the quantum information cost of~$\Pi$ 
for Alice with respect to~$\lambda$ as
\begin{equation}
\label{eqn-qic1}
\sum_{\textrm{odd } i \in [t]} \rI(X : Q_i \,|\, Y) \enspace,
\end{equation}
and similarly for Bob as 
\begin{equation}
\label{eqn-qic2}
\sum_{\textrm{even } i \in [t]} \rI( Y : P_i \,|\, X) \enspace.
\end{equation}
A similar definition has been considered by Jain, Radhakrishnan, 
and Sen~\cite{JainRS03b}.  
This appears to be a natural definition; it captures the amount of 
information about the other party's input that is not already contained 
in her state. It also allows us to relate quantum streaming algorithms for
$\dyck(2)$ that use small space, to two-party protocols for \AIndex\ with
small quantum information cost. (The reduction described in
Section~\ref{sec-streaming} extends to quantum algorithms with minor
modifications.) However, we are not able to prove an information cost
trade-off for \AIndex\ with this definition.

The tension between applicability and ease of analysis is rather acute
in our case. This leads us to consider the information contained in the
messages when the input qubits are initialized to an appropriate
\emph{superposition\/}. This information is in general more than 
that contained in the messages when we have the corresponding 
\emph{distribution\/} over inputs. The former measure may sometimes
capture the information revealed by a party in a quantum communication 
protocol more accurately (see, e.g., Ref.~\cite{JainRS09}). The 
resulting notion also seems to be necessary for the proof of the 
information cost trade-off we present. 

Defining quantum information cost with superpositions over inputs,
corresponding to arbitrary non-product distributions, comes with its 
own set of complications. A comprehensive discussion of such measures 
is beyond the scope of this article. We focus on distributions~$\lambda$
over the input space~$\cX \times \cY$ with~$\cY = \cY_1 \times \cY_2$, and
the following limited type of dependence. Let~$X,Y_1$ be independent 
random variables taking values in~$\cX, \cY_1$, respectively, 
and~$Y_2 = s(X,Y_1) \in \cY_2$, where~$s$ is some function of the 
first two random variables. Moreover, the function~$s$ is such
that the conditional random variables~$X|(Y_2 = v)$ and~$Y_1|(Y_2 = v)$
are also independent, for any~$v$ with~$\Pr[Y_2 = v] \neq 0$.
Then~$\lambda$ is the distribution of~$X Y_1 Y_2$. In other words, Alice
is given some input~$X$, Bob an independent input~$Y_1$, and also a joint
function~$Y_2 = s(X,Y_1)$ of the two. Moreover, their inputs~$X,Y_1$
remain independent when conditioned on any given value of~$Y_2$.
Such distributions include product distributions as well as
distributions for problems in which the two communicating parties may
share a portion of the input, as in the case of \AIndex. (The
correspondence for \AIndex\ is that~$X$ is uniformly distributed 
over~$\set{0,1}^n$, $Y_1$ is the index~$K$ that is uniformly distributed 
over~$[n]$, and~$Y_2 = s(X,Y_1) = X[1,K]$.)

The final point of difference between the notions of classical and quantum 
information cost we consider comes from the dependence described above in the
distribution~$\lambda$. Recall that under this distribution~$\lambda$,
Bob's input~$Y_1$ is independent of~$X$ and that Bob additionally gets~$Y_2 =
s(X,Y_1)$. In the classical case, Alice may have information about~$Y_2$
due to its dependence on~$X$, but does not have any information about~$Y_1$,
i.e., $\rI(X:Y_1) = 0$.  When the input registers are initialized with a 
superposition corresponding to~$\lambda$, however, Alice may \emph{gain\/}
information about Bob's input~$Y_1$ without any communication between 
the parties: we may have~$\rI(\hX:\hY_1) > 0$,
where~$\hX \hY_1 \hY_2$ are in state~$\sum_{x \in \cX, y \in \cY} 
\sqrt{\lambda(x,y)\,} \ket{x,y}$.

To illustrate this phenomenon, consider the following example. 
Let~$X$ be uniformly distributed over~$\set{0,1}^n$, $Y_1$ be an
index~$K$ that is uniformly distributed over~$[n]$, and~$Y_2 = X_K$,
i.e., the~$K$th bit of~$X$. We have~$\rI(X:Y_1) = 0$.
Let~$\hX \hY_1 \hY_2$ be initialized to the
state
\[
\frac{1}{\sqrt{n 2^n}} \sum_{x \in \set{0,1}^n, k \in [n]} \ket{x,k,x_k}
\enspace.
\]
Suppose we measure the qubits holding~$\hY_1$ in the 
basis~$(\ket{i})_{i \in [n]}$ and recover~$Y_1$. By monotonicity of mutual 
information under quantum 
operations~\cite[Theorem~11.15, p.~522]{NielsenC00}, we have~$\rI(\hX : \hY_1)
\geq \rI(\hX : Y_1)$.  The reduced state of~$\hX Y_1$ is
\[
\frac{1}{n} \sum_{k \in [n]} 
\density{u}^{\tensor (k-1)} \tensor \frac{\identity}{2} \tensor 
\density{u}^{\tensor (n-k)} \tensor \density{k}
\enspace,
\]
where~$\ket{u} = (\ket{0} + \ket{1})/\sqrt{2}$.  By conjugating~$\hX$ 
by the~$n$-qubit Hadamard operation, we see that the state is equivalent to
\[
\frac{1}{n} \sum_{k \in [n]}    
\density{0}^{\tensor (k-1)} \tensor \frac{\identity}{2} \tensor
\density{0}^{\tensor (n-k)} \tensor \density{k}
\enspace.
\]
A straightforward calculation now shows that~$\rI(\hX:Y_1) = \log_2 n$.
So~$\rI(\hX : \hY_1) \geq \log_2 n$, whereas~$\rI(X : Y_1) = 0$.

This phenomenon also occurs in the case of \AIndex, due to the 
prefix shared by the two parties. To quantify the information leaked 
\emph{by the protocol\/}, rather than \emph{the preparation of the initial
state in a superposition\/}, we view the protocol differently. We imagine 
that there is a single quantum register that carries the superposition 
corresponding to~$X$, and that Bob's unitary operations are controlled
appropriately by this register. In other words, his transformation
in the~$i$th round is of the form
\[
V_i \quad = \quad \sum_{x, y_1} \density{x} 
\tensor \density{y_1} \tensor V_{i, y_1 s(x,y_1)} \enspace,
\]
where the qubits holding~$x$ are with Alice.
Bob's information cost is then measured with respect to all the qubits 
with Alice.

We are now in a position to define the measure of quantum information
cost for two-party protocols that we analyze.
Let~$\lambda$ be a probability distribution over~$\cX \times \cY$ of the
type described above, and let~$\hX \hY_1$ denote the corresponding
superposition~$\sum_{x \in \cX, y_1 \in \cY_1} \sqrt{\lambda(x,y_1)\,}
\ket{x,y_1}$ over inputs.  Let~$\hX P_i Q_i \hY_{1}$ denote the joint 
state of Alice and Bob's input and workspace qubits \emph{immediately 
after\/} the~$i$th message is sent, in a protocol~$\Pi$ when we start 
with the input qubits in state~$\hX \hY_1$.
Note that the input qubits may get entangled with the message qubits during
the protocol. As the state of the input qubits we refer to will be clear 
from the context, we do not label it with the message number~$i$. 
The quantum information cost of~$\Pi$ for Bob with respect 
to~$\lambda$ is then defined as
\begin{eqnarray*}
\qic^\sB_{\lambda}(\Pi)
    & = & \sum_{\textrm{even } i \in [t]} \rI( \hY_{1} : \hX P_i) \enspace.
\end{eqnarray*}
In this cost, we measure the information about~$\hY_1$ contained in 
Alice's quantum state, while disregarding~$Y_2 = s(X,Y_1)$ (which is 
not available to Alice).

In Alice's cost, we would like to measure 
the information about~$\hX$ in Bob's quantum state, given access
to~$Y_2$. We model this as follows. We imagine an additional register 
that we label~$Y_2$. We copy~$s(X,Y_1)$ into this register and measure 
the qubits in the standard basis. The initial state of the
registers~$\hX \hY_1 Y_2$ is then
\[
\sum_{y_2 \in \cY_2}
\sum_{\substack{x,x' \in \cX \\ y_1, y_1' \in \cY_1 : \\
s(x,y_1) = s(x',y_1') = y_2}}
\sqrt{\lambda(x,y_1) \, \lambda(x',y_1')} \;
\ketbra{x,y_1}{x',y_1'} \tensor \density{y_2} \enspace.
\]
The joint state~$\hX P_i Q_i \hY_{1} Y_2$ of Alice and Bob's
input and workspace qubits, immediately after the~$i$th message is sent,
is correspondingly affected. We define Alice's information cost as
\begin{eqnarray*}
\qic^\sA_{\lambda}(\Pi) & = & \sum_{\textrm{odd } i \in [t]}
          \rI(\hX : Q_i \hY_{1} \,|\, Y_2) \enspace.
\end{eqnarray*}
The inclusion of the register holding~$Y_2$ precisely captures
the distribution of inputs in the communication protocol. The 
artificial construct described before, of substituting this with
suitable read-only access to Alice's input qubits (for executing
Bob's unitary transformations), however, is more appropriate for 
the proof of the quantum information cost trade-off.

The above notion corresponds to a hybrid of ``internal'' and 
``external information cost''~\cite{BarakBCR13}.  For product 
distributions (when~$Y_2$ is trivial), each term of this notion 
reduces precisely to the amount of (quantum) information 
available to a party about the other's input. 

In the rest of Section~\ref{sec-quantum}, we use a convention similar 
to the one above:
a symbol such as~$Z$ without a hat denotes the random variable resulting
from an imagined measurement, in the computational basis, of a 
sequence of qubits initialized to a superposition. The state of the 
qubits prior to the measurement is denoted by the symbol with a hat, 
e.g.,~$\hZ$.

Measuring any part of a quantum system in general affects the state of
the remaining qubits. Thus the symbol~$\hX$ used in the
expressions for Alice's and Bob's information cost denotes potentially 
different states. In the analysis that we present for \AIndex, we 
imagine measurements only of parts of Alice's and Bob's inputs in the 
computational basis. In that case, we denote the resulting state of 
the qubits without a hat. Thus the state we mean will be clear from 
the context.

\subsection{The quantum information cost trade-off}
\label{sec-qlb}

In this section, we derive an analogue of the information trade-off
result established in Section~\ref{sec-main} for quantum communication
protocols for \AIndex.

We first specialize the notion of quantum information cost to the
\AIndex\ function~$f_n$, and simplify it further. This allows us derive
a stronger information cost trade-off than with the original definition.
Let~$(X, K, B)$
be random variables distributed according to~$\mu$, the uniform
distribution over~$\{0,1\}^n \times [n] \times \set{0,1}$. Let~$\mu_0$
denote the distribution~$\mu$ conditioned upon~$X_K = B$, i.e., when
the inputs are chosen uniformly from the set of~$0$s of~$f_n$. We are
interested in the quantum information cost of a protocol~$\Pi$ for
\AIndex\ under the distribution~$\mu_0$, for the two parties.

As explained in Section~\ref{sec-qcommn}, we adopt the following 
convention with respect to the inputs for \AIndex.  Alice is given 
the input~$x$. We imagine that Bob is given~$k,b$, and
\emph{access\/} to the prefix~$x[1,k-1]$, rather than a copy of these
bits. When we restrict to the distribution~$\mu_0$, we assume he has
read-only access to~$x[1,k]$. This means that in any round~$i$ of the
protocol in which Bob plays, his local unitary operation~$V_i$ is
controlled by the qubits with Alice that hold the prefix.
It is important to bear in mind the qubits on which the unitary 
operations of the protocol act non-trivially, i.e., do not equal 
the identity.  In particular, in Lemma~\ref{lem:key}, we use the 
commutativity of the unitary operations used
in the protocol and the corresponding unitary operations given
by Lemmata~\ref{thm-pairs-quantum} and~\ref{thm-flip-quantum}. See, for
example, the paragraph before Eq.~(\ref{eqn-flip-r1}).

Suppose we have a quantum protocol~$\Pi$ for \AIndex\ with a total
of~$t$ messages.  Without loss of generality (see Section~\ref{sec-qcommn}),
we assume that Alice sends the first message, and alternates with Bob
thereafter.

Let~$\hX P_i Q_i \hK \hB$ denote the joint state of Alice and Bob's workspace in the
protocol~$\Pi$ immediately after the~$i$th message is sent, when we 
start with uniform
superpositions~$\hX$ over strings~$x \in \set{0,1}^n$, $\hK$ over~$[n]$,
and~$\hB$ over~$\set{0,1}$ (this corresponds to distribution~$\mu$). 
Let~$\hX^0 P_i^0 Q_i^0 \hK^0 \hB^0$ denote the analogous joint state corresponding
to~$\mu_0$, where we assume that Bob is given read-only access to the
register containing~$x_k$, rather than a copy of this bit. The quantum
information cost of~$\Pi$ for Alice and Bob with respect to~$\mu_0$ is
then
\begin{eqnarray*}
\qic^\sA_{\mu_0}(\Pi)
    & = & \sum_{\textrm{odd } i \in [t]}
          \rI(\hX^0[K+1,n] : Q_i^0 \hK^0 \,|\, X[1,K] ) \enspace, 
          \qquad \textrm{and} \\
\qic^\sB_{\mu_0}(\Pi)
    & = & \sum_{\textrm{even } i \in [t]} \rI( \hK^0 : \hX^0 P_i^0) \enspace.
\end{eqnarray*}
Due to the
\suppress{
concavity of von Neumann entropy~\cite[Section~11.3.5, p.~516]{NielsenC00},
and the
}
monotonicity of mutual information under quantum
operations~\cite[Theorem~11.15, p.~522]{NielsenC00}, for each~$i = 1,
\dotsc, t$ we have
\begin{eqnarray*}
\rI(X : Q_i^0 \,|\, X[1,K] ) 
    & \le & \rI( \hX^0[K+1,n] : Q_i^0 \hK^0 \,|\, X[1,K] )
            \enspace, \qquad \textrm{and} \\
\rI( K : \hX^0 P_i^0) & \le & \rI( \hK^0 : \hX^0 P_i^0) \enspace,
\end{eqnarray*}
where the symbols without a hat denote random variables resulting from
an imagined measurement of the corresponding qubits in the computational 
basis.  (We drop the superscript~`$0$' on these random variables, as their 
marginals are the same as under the distribution~$\mu$.)
The trade-off we prove also holds for the potentially smaller 
quantities on the left side above. In order to state the theorem in the
strongest possible terms, we define another measure of information cost
as follows:
\begin{eqnarray*}
\tqic^\sA_{\mu_0}(\Pi)
    & = & \sum_{\textrm{odd } i \in [t]}
          \rI(X : Q_i^0 \,|\, X[1,K] ) \enspace, \qquad \textrm{and} \\
\tqic^\sB_{\mu_0}(\Pi)
    & = & \sum_{\textrm{even } i \in [t]} \rI( K : \hX^0 P_i^0) \enspace.
\end{eqnarray*}

The intuition behind the lower bound on quantum information cost is
the same as that in the classical case. Namely, starting from an input
pair on which the function evaluates to~$0$, if the information cost
of any one party is low and we carefully change her input, the other
party's share of the state does not change much. Assume for simplicity
that Alice produces the output of the protocol. We show that
even when we simultaneously change both parts of the input, resulting
in a~$1$-input of the function, the perturbation to Alice's final
state is also correspondingly small. This implies that
the two information costs cannot be small simultaneously. For more
intuition into the main lemmata in this proof, we refer the reader to the
analogous steps in the classical case. In the final piece of the argument 
for the quantum case, the Local Transition Theorem and a hybrid
argument take the place of the Cut-and-Paste Lemma. Unlike the latter,
these are applied on a message-by-message basis, \emph{{\`a} la\/}
Jain, Radhakrishnan, and Sen~\cite{JainRS03b}, and leads to a
dependence of the information cost trade-off on the number of messages
in the protocol. 

The next theorem executes this argument for even~$n$. A similar result
also holds for odd~$n$, and may be inferred from the proof for the
even case. As explained in the previous section, the assumption that
Alice sends the first message is not necessary. 

\begin{theorem}
\label{thm-main-quantum}
Let~$\Pi$ be any quantum two-party communication protocol for the
\AIndex\ function~$f_n$ with~$n$ even, Alice starting
and alternating with Bob for a total of~$t \ge 1$ messages.
If~$\Pi$ makes error at most~$\eps \in [0, 1/4]$ on the
uniform distribution~$\mu$ over inputs, then
\[
2 \left[ \frac{\tqic_{\mu_0}^\sA(\Pi)}{n} \right]^{1/2}
+ \left[ 2\cdot \tqic_{\mu_0}^\sB(\Pi) \right]^{1/2}
    \quad \geq \quad \frac{1-4\eps}{4\sqrt{\kappa\, t}} \enspace,
\]
where~$\mu_0$ is the uniform distribution over~$f_n^{-1}(0)$.
\end{theorem}
\begin{proof}
Consider a protocol~$\Pi$ as in the statement of the theorem.  Let the
inputs be given by random variables~$X,K,B$, drawn from the
distribution~$\mu$, let~$d \eqdef \tqic_{\mu_0}^\sA(\Pi)/n$, and let~$c
\eqdef \tqic_{\mu_0}^\sB(\Pi)$.

Let~$\hX P_i Q_i K B$ be the joint state of the registers used in the
protocol, when the inputs are initialized with a uniform
superposition~$\hX$ over~$x \in \set{0,1}^n$ and random
variables~$K,B$, immediately after the~$i$th message in the
protocol. Let~$d_i = \tfrac{1}{n} \, \rI( X : Q_i^0 \,|\, X[1,K])$ for
odd~$i \in [t]$, and $c_i = \rI( K : \hX^0 P_i^0)$ for even~$i \in
[t]$. So~$d = \sum_{\textrm{odd } i \in [t]} d_i$ and~$c =
\sum_{\textrm{even } i \in [t]} c_i$.

We prove the theorem assuming that Alice computes the output of the 
protocol, i.e., $t$ is even. The proof when Bob computes the output is
similar; we point out the main differences along the way.
If~$t$ is even, we show that the state~$X P_t^0$ is close in
trace distance to the state~$X P_t^1$, where~$X P_t^1$
denotes the reduced state~$X P_t$ conditioned on the function value
being~$1$, i.e., when~$B = {\bar{X}_K}$. (Note that~$X$ is the classical
random variable corresponding to the superposition~$\hX$.)
\begin{lemma}
\label{thm-close-quantum}
For even~$t$,
$\trnorm{X P_t^0 - X P_t^1} 
    \quad \leq \quad 1 + 4 \sqrt{\kappa\, t}
                     \left[ 2 \sqrt{d} + \sqrt{2 c} \right]$, ~~
where~$\kappa = \frac{\ln 2}{2}$.
\end{lemma}
If~$t$ is odd, i.e., Bob computes the output of the protocol, we
show the same bound on
\[
\trnorm{Q_t^0\, X[1,K] - Q_t^1\, X[1,K-1]\, \bar{X}_K} \enspace.
\]
Since the protocol identifies the two states~$X P_t^0$ and~$X P_t^1$, with
average error~$\eps$, and trace distance is monotonic under quantum
operations~\cite[Theorem~9.2, p.~406]{NielsenC00}, we have
\[
\trnorm{X P_t^0 - X P_t^1} \quad \geq \quad 2(1-2\eps) \enspace.
\]
The theorem follows.
\end{proof}

We now prove the core of the theorem, i.e., that if Alice computes the
output, her final state
for the~$0$ and~$1$ inputs are close to each other in distribution.

\begin{proofof}{Lemma~\ref{thm-close-quantum}} 
{}
When we wish to explicitly write a state, say~$P_i$, as a function of
the inputs to Alice and Bob, say $x$ and $x[1,k-1], b$ respectively,
we write it as $P_i(x ; x[1,k-1],b)$. If~$b = x_k$, we write Bob's
input as~$x[1,k]$.

As before, for any~$x \in \set{0,1}^n$ and~$i \in [n]$, we
let~$x^{(i)}$ denote the string that equals~$x$ in all coordinates
except at the~$i$th. Note that~$P_t^1 = P_t(X ; X[1,K-1],\bar{X}_K)$
is the same mixed state as~$P_t( X^{(K)} ; X[1,K])$, since~$X$
and~$X^{(K)}$ are identically distributed. Thus, our goal is to bound
\[
\trnorm{ X P_t(X ; X[1,K]) - X^{(K)} P_t( X^{(K)} ; X[1,K]) }
\enspace.
\]

For reasons similar to those the classical case and new ones arising
from our proof (an explanation for which is included below), we consider 
the trace distance between the first
term above with~$K \in [n/2]$ and the second term with~$K \in [n] -
[n/2]$. (Recall that in the classical case, we restricted ourselves
to~$K \in [n]-[n/2]$ in both terms.)  Let~$J$ be uniformly and
independently distributed in~$[n/2]$, and let~$L$ be uniformly and
independently distributed in~$[n]-[n/2]$. Then
\begin{eqnarray}
\nonumber
\lefteqn{ \trnorm{ X P_t(X ; X[1,K]) - X^{(K)} P_t( X^{(K)} ; X[1,K]) } } \\
\nonumber
    & = & \left\| \frac{1}{2} \left( X P_t(X ; X[1,J]) + X P_t(X ; X[1,L])
          \right) \right. \\
\nonumber
    &   & \mbox{} - \left. \frac{1}{2} \left( X^{(J)} P_t( X^{(J)} ; X[1,J])
          + X^{(L)} P_t( X^{(L)} ; X[1,L]) \right) \right\|_{\mathrm{tr}} \\
\nonumber
    & \leq & 1 + \frac{1}{2} 
             \norm{ X P_t(X ; X[1,J]) - X^{(L)} P_t( X^{(L)} ; X[1,L]) }
             \\
\label{eqn-bound-quantum}
    & = & 1 + \frac{1}{2} 
             \norm{ X^{(L)} P_t(X^{(L)} ; X[1,J])
                    - X^{(L)} P_t( X^{(L)} ; X[1,L]) }
             \enspace,
\end{eqnarray}
where we use the fact that~$X$ and~$X^{(L)}$ are identically
distributed, even given the prefix~$X[1,J]$, and that the states~$X
P_t(X ; X[1,J])$ and~$X^{(L)} P_t(X^{(L)} ; X[1,J])$ are therefore identical.
So it suffices to bound the RHS above. If~$t$ is odd, we instead bound
\begin{eqnarray}
\nonumber
\lefteqn{ \trnorm{ Q_t(X ; X[1,K]) X[1,K] 
                   - Q_t( X^{(K)} ; X[1,K]) X[1,K]} } \\
\label{eqn-bound-quantum2}
    & \leq & 1 + \frac{1}{2} 
             \trnorm{ Q_t(X ; X[1,L]) X[1,L] 
                      - Q_t( X^{(L)} ; X[1,L]) X[1,L]}
             \enspace.
\end{eqnarray}
The expression for odd~$t$, Eq.~(\ref{eqn-bound-quantum2}), is similar 
to the one we had in the classical 
case: we focus on the case~$K \in [n] - [n/2]$ alone.

For every~$j \in [n/2], l \in [n] - [n/2]$ and $z \in \set{0,1}^l$, we
consider four runs of the protocol~$\Pi$. The inputs to Alice and Bob
in the four runs are summarized in the table below. Only the first~$l$
bits of Alice's input are specified. In all four runs, the
last~$(n-l)$ input bits of Alice are initialized to a uniform
superposition over all~$(n-l)$-bit strings. The final column gives the
notation for the (pure) state corresponding to the
registers~$\hX[l+1,n]\, P_i Q_i$, which constitute the last~$(n-l)$
inputs bits of Alice, her workspace, and that of Bob, immediately
after the~$i$th message has been sent, $i \in [t]$.

\begin{center}
\begin{tabular}{|c|l|l|l|}
\hline
Run & Alice's input $x[1,l]$ & Bob's input $k,x[1,k-1],b$ & State \\
\hline
00 & $z$ & $j, z[1,j-1], z_j$ & $\ket{\phi_i(z,j)}$ \\
01 & $z$ & $l, z[1,l-1], z_l$ & $\ket{\phi_i(z,l)}$ \\
10 & $z^{(l)}$ & $j, z[1,j-1], z_j$ & $\ket{\phi_i(z^{(l)},j)}$ \\
11 & $z^{(l)}$ & $l, z[1,l-1], z_l$ & $\ket{\phi_i(z^{(l)},l)}$ \\
\hline
\end{tabular}
\end{center}
The two bits in the ``Run'' column indicate whether Alice's~$l$th bit 
has been flipped, and whether we have switched~$j$ to~$l$. A~``1''
indicates a switch. Note that for the first three kinds of inputs, the
function value is~$0$, and for the last it is~$1$.

When Bob's information cost is low, it follows that the final state on 
inputs of type~``00'' is close to the final state on inputs of type~``01''
(Lemma~\ref{thm-pairs-quantum}). We show a similar closeness between 
the final state on inputs of type~``10'' and that on inputs of type~``11''.
This explains the choice made in Eq.~(\ref{eqn-bound-quantum}) when 
Alice produces the output of the protocol. For similar reasons, when Bob
produces the output of the protocol, we compare the final state of the 
protocol on inputs of type~``01'' with that on inputs of type~``11'',
as in Eq.~(\ref{eqn-bound-quantum2}).

As the first step,
we compare the intermediate protocol states in the above four runs,
when we flip the~$l$th input bit of Alice, and when we switch Bob's
input from~$j$ to~$l$ (along with the corresponding prefix). We show
that the switch results in a perturbation to reduced state of the
other party that is related to the information contained about the bit
or the index (as in the classical case). To quantify this
perturbation, define
\begin{eqnarray*}
h_i(j,l,z) & = & \fh{ Q_i( z X[l+1,n] ; z[1,j]) }
                    { Q_i( z^{(l)} X[l+1,n] ; z[1,j]) }
                 \enspace,
\end{eqnarray*}
for every odd~$i \in [t]$. This is the perturbation in Bob's reduced state
when we flip the~$l$th bit of Alice input, when Bob has index~$j$. 
\suppress{
  The state~$Q_i$ above results from running the protocol with a
  superposition over strings. However, Since tracing out the register
  with~$\hX[l+1,n]$ has the same effect as measuring it in the
  standard basis, we write~$Q_i$ as a function of the random
  variable~$X$.
}
Define
\begin{eqnarray*}
h_i(j,l,z) & = & \fh{ \hX[l+1,n] \, P_i(z \hX[l+1,n]; z[1,j]) }
                    { \hX[l+1,n] \, P_i(z \hX[l+1,n]; z[1,l]) }
                 \enspace,
\end{eqnarray*}
for every even~$i \in [t]$. This is the perturbation in Alice's reduced 
state when we switch Bob's index from~$j$ to~$l$.  In the above states,
$P_i$ is entangled with the qubits holding~$\hX$, and is written as a 
function of~$\hX[l+1,n]$ to emphasize this.

The number of qubits Alice and Bob have during the protocol changes
with every message. To maintain simplicity of notation, we denote the
identity operator in any round on the register holding~$\hX[l+1,n]$
and Alice's workspace qubits by~$\identity_\sA$ and the identity
operator on Bob's workspace qubits by~$\identity_\sB$.

We begin by showing that changing Bob's input alone from~$j$ to~$l$
while keeping Alice's input fixed at~$z \hX[l+1,n]$, does not perturb
Alice's reduced state in any round of communication by much, provided
the corresponding information cost of Bob is small. By the Local
Transition Theorem, we then see that Bob may apply a unitary operation
to his qubits alone to bring the protocol states close to each other.
\newtheorem*{thm-pairs-quantum}{Lemma~\ref{thm-pairs-quantum}}
\begin{thm-pairs-quantum}
For every even~$i \in [t]$, there is a unitary operator~$U_i$
that depends upon~$j,l,z$, acts on Bob's workspace qubits alone (i.e.,
on the register holding state~$Q_i$), and is such that
\[
\fh{ \, (\identity_\sA \tensor U_i) \, \ket{\phi_i(z,j)}}
   { \ket{\phi_i(z,l)}} 
     \quad = \quad h_i(j,l,z) \enspace.
\]
Moreover,
\[
\expct_{(j',l',z') \leftarrow (J,L,X[1,L])} \; h_i(j',l',z')
\quad \leq \quad \sqrt{8 \kappa\, c_i} \enspace.
\]
\end{thm-pairs-quantum}
The proof is presented later in this section.

Next, we show that if the information cost of Alice is small, Bob's
state~$Q_i^0$ does not carry much information about~$X$, even given a
prefix. Therefore, flipping a bit outside the prefix does not perturb
Bob's state by much, and there is a unitary operation on Alice's
qubits which brings the joint states close to each other.
\newtheorem*{thm-flip-quantum}{Lemma~\ref{thm-flip-quantum}}
\begin{thm-flip-quantum}
For every odd~$i \in [t]$, there is a unitary operator~$U_i$ that
depends upon~$j,l,z$, acts on the qubits holding~$\hX[l+1,n]$ and
Alice's workspace qubits (the register holding state~$P_i$), and is
such that
\[
\fh{ \, (U_i \tensor \identity_\sB) \, \ket{\phi_i(z,j)}}
   { \ket{\phi_i(z^{(l)},j)}} 
     \quad = \quad h_i(j,l,z) \enspace.
\]
Moreover,
\[
\expct_{(j',l',z') \leftarrow (J,L,X[1,L])} \; h_i(j',l',z') 
    \quad \leq \quad 4 \sqrt{\kappa\, d_i} \enspace.
\]
\end{thm-flip-quantum}
This is proven later in the section.

There is no quantum counterpart to the Cut-and-Paste lemma, so that
unlike in the classical case, the above two lemmata are by themselves
not sufficient to conclude the theorem. Instead, we combine these with
a hybrid argument to show that switching from chosen~$0$-inputs of \AIndex\ of
the type~``10'' (as defined above) to corresponding~$1$-inputs of type~``11''
does not affect the final state by ``much''.
\newtheorem*{lem:key}{Lemma~\ref{lem:key}}
\begin{lem:key}
Let~$(U_i)_{i \in [t]}$, be the unitary operators given by
Lemmata~\ref{thm-pairs-quantum} and~\ref{thm-flip-quantum}.
For every odd~$r \in [t]$, 
\begin{eqnarray*}
\fh{ (U_r \otimes \identity_\sB) \, \ket{\phi_r(z,l)} }
             { \ket{\phi_r(z^{(l)}, l)} }
    & \leq  & h_r(j,l,z) + 2 \sum_{i = 1}^{r-1}  h_i(j,l,z)
            \enspace .
\end{eqnarray*}
For every even~$r \in [t]$,
\begin{eqnarray*}
\fh { ( \identity_\sA \otimes U_r) \ket{\phi_r(z^{(l)}, j)} }
        { \ket{\phi_r(z^{(l)}, l)} }
    & \leq & h_r(j,l,z) + 2 \sum_{i = 1}^{r-1}  h_i(j,l,z)
            \enspace .
\end{eqnarray*}
\end{lem:key}
This is proved later in this section.

Recall that~$t$ is even. We have
\suppress{
By the Triangle Inequality, the monotonicity of the trace distance under
quantum operations~\cite[Theorem~9.2, p.~406]{NielsenC00},
the relationship between trace and Bures distance 
(Proposition~\ref{lem-trbures}), Lemmata~\ref{lem:key},
\ref{thm-pairs-quantum} and~\ref{thm-flip-quantum},
}
\begin{eqnarray*}
\lefteqn{
\trnorm{ X^{(L)} P_t(X^{(L)} ; X[1,J]) - X^{(L)} P_t( X^{(L)} ; X[1,L]) }
} \\
    & \leq & \expct_{(j,l,z) \leftarrow (J,L,X[1,L])}
             \trnorm{ X[l+1,n]\, P_t(z^{(l)} X[l+1,n] ; z[1,j]) 
                      - X[l+1,n]\, P_t( z^{(l)} X[l+1,n] ; z[1,l]) } \\
    &      & \qquad \textrm{(by the Triangle Inequality)} \\
    & \leq & \expct_{(j,l,z) \leftarrow (J,L,X[1,L])}
             \trnorm{ \hX[l+1,n]\, P_t(z^{(l)} \hX[l+1,n] ; z[1,j]) 
                      - \hX[l+1,n]\, P_t( z^{(l)} \hX[l+1,n] ; z[1,l]) } \\
    &      & \qquad \textrm{(by the monotonicity of trace
             distance under quantum
             operations~\cite[Theorem~9.2, p.~406]{NielsenC00})} \\
    & \leq & 2\sqrt{2} ~ \expct_{(j,l,z) \leftarrow (J,L,X[1,L])} \;
             \fh{ \hX[l+1,n]\, P_t(z^{(l)} \hX[l+1,n] ; z[1,j]) }
                { \hX[l+1,n]\, P_t(z^{(l)} \hX[l+1,n] ; z[1,l]) } \\
    &      & \qquad \textrm{(by Proposition~\ref{lem-trbures})} \\
    & \leq &  2\sqrt{2} ~ \expct_{(j,l,z) \leftarrow (J,L,X[1,L])} \;
              \fh{ ( \identity_\sA \otimes U_t) \ket{\phi_t(z^{(l)}, j)} }
                 { \ket{\phi_t(z^{(l)}, l)} } \\
    &      & \qquad \textrm{(by monotonicity of Bures distance under quantum
operations~\cite[Theorem~9.6, p.~414]{NielsenC00})} \\
    & \leq &  4 \sqrt{2} ~ \expct_{(j,l,z) \leftarrow (J,L,X[1,L])}
              \sum_{i = 1}^{t}  h_i(j,l,z) 
             \qquad \textrm{(by Lemma~\ref{lem:key})}  \\
    & \leq & 4 \sqrt{2} \left[
             \sum_{\textrm{odd } i \in [t]} 4 \sqrt{\kappa\, d_i}
             + \sum_{\textrm{even } i \in [t]} 2 \sqrt{2 \kappa\, c_i}
             \right] 
             \qquad \textrm{(by Lemmata~\ref{thm-pairs-quantum} 
             and~\ref{thm-flip-quantum})} \\
    & \leq & 8 \sqrt{\kappa\, t} \left[ 2 \sqrt{d} + \sqrt{2 c} \right]
             \enspace.
             \qquad \textrm{(by the Jensen Inequality)}
\end{eqnarray*}
In deriving the fourth inequality above, we used the fact that the states
here are purification of the states in the previous inequality.
This gives us a bound on the RHS of Eq.~(\ref{eqn-bound-quantum}), and
concludes the proof of Lemma~\ref{thm-close-quantum}.
\end{proofof}

We turn to the deferred proofs.
\begin{lemma}
\label{thm-pairs-quantum}
For every even~$i \in [t]$, there is a unitary operator~$U_i$
that depends upon~$j,l,z$, acts on Bob's workspace qubits alone (i.e.,
on the register holding state~$Q_i$), and is such that
\[
\fh{ \, (\identity_\sA \tensor U_i) \, \ket{\phi_i(z,j)}}
   { \ket{\phi_i(z,l)}} 
     \quad = \quad h_i(j,l,z) \enspace.
\]
Moreover,
\[
\expct_{(j',l',z') \leftarrow (J,L,X[1,L])} \; h_i(j',l',z')
\quad \leq \quad \sqrt{8 \kappa\, c_i} \enspace.
\]
\end{lemma}
\begin{proof}
Note that~$\hX[l+1,n]\, P_i(z \hX[l+1,n]; z[1,k])$ for~$k \leq l$ is
the reduced state of~$\ket{\phi(z,k)}$ with Bob's workspace (i.e., the
register holding state~$Q_i$) traced out.  By the Local Transition
Theorem, Proposition~\ref{fact:localtrans}, there is a unitary
operator~$U_i$ that depends upon~$j,l,z$, acts on Bob's workspace
qubits alone, and is such that
\[
\fh{ \, (\identity_\sA \tensor U_i) \, \ket{\phi_i(z,j)}}
   { \ket{\phi_i(z,l)}}
     \quad = \quad h_i(j,l,z) \enspace.
\]

We show that this distance is bounded on average.
Consider the quantum state~$\hX \tP_i$ which is the reduced state of all
quantum registers except Bob's workspace and his input~$K$. We denote
by~$\hX P_i(\hX; \hX[1,k])$ this state for a fixed index~$k$, so that
\[
\hX \tP_i \quad = \quad \frac{1}{n} \sum_{k = 1}^{n} 
                        \hX P_i(\hX; \hX[1,k]) \enspace.
\]
By the Average Encoding Theorem, Proposition~\ref{thm-avg-quantum}, 
\[
\expct_{k \leftarrow K} \; \fh{ \hX P_i(\hX \,;\, \hX[1,k])}{
  \hX \tP_i}^2 \quad \leq \quad \kappa\, c_i \enspace,
\]
where~$\kappa = \frac{\ln 2}{2}$.
An immediate consequence is that
\begin{eqnarray*}
\expct_{j' \leftarrow J} \; \fh{ \hX P_i( \hX \,;\, \hX[1,j']) }{ \hX \tP_i}^2
    & \leq & 2 \, \kappa \, c_i \enspace, \qquad \text{and} \\
\expct_{l' \leftarrow L} \; \fh{ \hX P_i( \hX \,;\, \hX[1,l']) }{ \hX \tP_i}^2
    & \leq & 2 \, \kappa \, c_i \enspace.
\end{eqnarray*}
By the Triangle Inequality, for any~$j' \in [n/2]$, $l' \in [n]-[n/2]$,
\begin{eqnarray*}
\lefteqn{ \fh{ \hX P_i(\hX \,;\, \hX[1,j'])}{  \hX P_i(\hX \,;\, \hX[1,l'])}^2} \\
& \leq &  \left(\fh{ \hX P_i(\hX \,;\, \hX[1,j'])}{ \hX \tP_i}
              + \fh{ \hX P_i(\hX \,;\, \hX[1,l'])}{ \hX \tP_i} \right)^2 \\
& \leq &  2 \, \fh{ \hX P_i(\hX \,;\, \hX[1,j'])}{  \hX \tP_i}^2
          + 2 \, \fh{ \hX P_i(\hX \,;\, \hX[1,l'])}{ \hX \tP_i}^2 \enspace .
\end{eqnarray*}
Since Bures distance is monotonic under quantum 
operations~\cite[Theorem~9.6, p.~414]{NielsenC00}, measuring the
first~$l'$ qubits of~$\hX$ yields
\begin{align*}
\lefteqn{ {\mathfrak h}\! \left( X[1,l']\, \hX[l'+1,n]\, 
               P_i(X[1,l']\, \hX[l'+1,n] \,;\, X[1,j']) \, , \right. }
          \\
    & \qquad  \left. X[1,l']\, \hX[l'+1,n]\, 
               P_i(X[1,l']\, \hX[l'+1,n] \,;\, X[1,l']) \right)^2  \\
    & \leq \quad 2 \, \fh{ \hX P_i(X \,;\, X[1,j'])}{ \hX \tP_i}^2
          + 2 \, \fh{ \hX P_i(X \,;\, X[1,l'])}{ \hX \tP_i}^2 \enspace ,
\end{align*}
where~$X[1,l']$ denotes the classical random variable resulting from 
the measurement of~$\hX[1,l']$.
Moreover, by Proposition~\ref{lem-convexity}, the left hand side above
is equal to
\[
    \expct_{z' \leftarrow X[1,l']} \;
          \fh{ \hX[l'+1,n]\, P_i(z' \hX[l'+1,n] \,;\, z'[1,j'])}
             { \hX[l'+1,n]\, P_i(z' \hX[l'+1,n] \,;\, z'[1,l'])}^2
          \enspace .
\]
Taking expectation over~$(j',l') \leftarrow (J,L)$, and invoking the
Jensen inequality, we get the claimed bound.
\end{proof}

\begin{lemma}
\label{thm-flip-quantum}
For every odd~$i \in [t]$, there is a unitary operator~$U_i$ that
depends upon~$j,l,z$, acts on the qubits holding~$\hX[l+1,n]$ and
Alice's workspace qubits (the register holding state~$P_i$), and is
such that
\[
\fh{ \, (U_i \tensor \identity_\sB) \, \ket{\phi_i(z,j)}}
   { \ket{\phi_i(z^{(l)},j)}} 
     \quad = \quad h_i(j,l,z) \enspace.
\]
Moreover,
\[
\expct_{(j',l',z') \leftarrow (J,L,X[1,L])} \; h_i(j',l',z') 
    \quad \leq \quad 4 \sqrt{\kappa\, d_i} \enspace.
\]
\end{lemma}
\begin{proof}
{}
Note that~$Q_i(zX[l+1,n]; z[1,k])$ for~$k \leq l$ is
the reduced state of~$\ket{\phi(z,k)}$ with the register holding~$\hX$
and Alice's workspace (the register holding state~$P_i$) traced out.
By the Local Transition Theorem, Proposition~\ref{fact:localtrans},
there is a unitary operator~$U_i$ that depends upon~$j,l,z$, acts on
the registers holding~$\hX[l+1,n]\, P_i$ alone, and is such that
\[
\fh{ \, (U_i \tensor \identity_\sB) \, \ket{\phi_i(z,j)}}
   { \ket{\phi_i(z^{(l)},j)}}
     \quad = \quad h_i(j,l,z) \enspace.
\]

Since~$Q_i^0 = Q_i(X\, ; \, X[1,K])$, we have
\begin{eqnarray}
\label{eqn-J-quantum}
\rI(X : Q_i(X\,;\, X[1,J]) \,|\, X[1,J]) 
    & \leq & 2 \; \rI(X : Q_i^0 \,|\, X[1,K]) \quad = \quad 2 d_i n \enspace.  
\end{eqnarray}
Fix~$j' \in [n/2]$ and~$z'' \in \set{0,1}^{j'}$. By the Chain Rule,
Proposition~\ref{fact-chain-quantum},
\begin{eqnarray}
\nonumber
\lefteqn{ \rI( X[j'+1,n] : Q_i(z'' X[j'+1,n] \,;\, z'') ) } \\
\nonumber
    & = & \sum_{l' = j'+1}^n 
          \rI( X_{l'} : Q_i( z'' X[j'+1,n] \,;\, z'')
               \,|\, X[j'+1,l'-1]) \\
\label{eqn-L-quantum}
    & \geq & \sum_{l' = n/2 +1}^n 
          \rI( X_{l'} : Q_i( z'' X[j'+1,n] \,;\, z'')
               \,|\, X[j'+1,l'-1]) 
          \enspace.
\end{eqnarray}
Moreover by the Triangle Inequality,
and the Average Encoding Theorem (Proposition~\ref{thm-avg-quantum}),
for any given $l' \in [n]-[n/2]$ and~$z'
\in \set{0,1}^{l'}$,
\begin{align}
\label{eqn-flip-quantum}
\nonumber
& \fh{  Q_i(z' X[l'+1,n] \,;\, z'[1,j'])}
     { Q_i(z'^{(l')} X[l'+1,n] \,;\, z'[1,j']) } \\
\nonumber
& \leq \quad \fh{  Q_i(z' X[l'+1,n] \,;\, z'[1,j'])}
     { Q_i(z'[1,l'-1] \, X_{l'} X[l'+1,n] \,;\, z'[1,j']) } \\
\nonumber
&  \qquad \mbox{} + \fh{  Q_i(z'^{(l')} X[l'+1,n] \,;\, z'[1,j'])}
     { Q_i(z'[1,l'-1] \, X_{l'} X[l'+1,n] \,;\, z'[1,j']) } \\
& \leq \quad \left[ \; 4 \kappa \;
       \rI( X_{l'} : Q_i(z'[1,l'-1]\, X_{l'} \, X[l'+1,n] \,;\, z'[1,j']))
        \; \right]^{1/2}    \enspace .
\end{align}
Combining Eqs.~(\ref{eqn-J-quantum}), (\ref{eqn-L-quantum}),
and~(\ref{eqn-flip-quantum}), we get
\begin{align*}
& \expct_{(j',l',z') \leftarrow (J,L,X[1,L])} \;
\fh{ Q_i(z' X[l'+1,n] \,;\, z'[1,j'])} 
   { Q_i(z'^{(l')} X[l'+1,n] \,;\, z'[1,j'])}^2 \\
& \leq \quad  4 \kappa \; \expct_{(j',l',z') \leftarrow (J,L,X[1,L])} \, 
              \rI( X_{l'} : Q_i(z'[1,l'-1]\, X_{l'} X[l'+1,n] \,;\, z'[1,j']))
    \\
& = \quad  4 \kappa \; \expct_{(j',l',z'') \leftarrow (J,L,X[1,J])} \, 
           \rI( X_{l'} : Q_i(z''\, X[j'+1,n] \,;\, z'')
               \, | \, X[j'+1,l'-1]) \\
& \leq \quad \frac{8\kappa}{n} \; \rI( X : Q_i(X \,;\, X[1,J]) \, | \, X[1,J]) 
\quad \leq \quad 16 \kappa \; d_i \enspace,
\end{align*}
as claimed.
\end{proof}

\begin{lemma}
\label{lem:key}
Let~$(U_i)_{i \in [t]}$, be the unitary operators given by
Lemmata~\ref{thm-pairs-quantum} and~\ref{thm-flip-quantum}.
For every odd~$r \in [t]$, 
\begin{eqnarray*}
\fh{ (U_r \otimes \identity_\sB) \, \ket{\phi_r(z,l)} }
             { \ket{\phi_r(z^{(l)}, l)} }
    & \leq  & h_r(j,l,z) + 2 \sum_{i = 1}^{r-1}  h_i(j,l,z)
            \enspace .
\end{eqnarray*}
For every even~$r \in [t]$,
\begin{eqnarray*}
\fh { ( \identity_\sA \otimes U_r) \ket{\phi_r(z^{(l)}, j)} }
        { \ket{\phi_r(z^{(l)}, l)} }
    & \leq & h_r(j,l,z) + 2 \sum_{i = 1}^{r-1}  h_i(j,l,z)
            \enspace .
\end{eqnarray*}
\end{lemma}
\begin{proof}
We prove the lemma by induction over~$r \in [t]$. The base case is~$r =
1$. By the convention we have adopted, Alice sends the first message.
Since the joint state immediately after the first message is independent
of Bob's input, we have
\[
\ket{\phi_1(z,l)} \quad =  \quad \ket{\phi_1(z,j)}
\qquad \textrm{and} \qquad
\ket{\phi_1(z^{(l)},l)} \quad =  \quad \ket{\phi_1(z^{(l)},j)} \enspace.
\]
That is, the state on the input of type~``01'' equals that on the
input of type~``00''. The same holds for inputs of type~``11'' 
and~``10''. Along with Lemma~\ref{thm-flip-quantum} we get
\begin{eqnarray*}
\lefteqn{
\fh{ (U_1 \otimes \identity_\sB) \, \ket{\phi_1(z,l)} } 
{ \ket{\phi_1(z^{(l)}, l)} }
} \\
    & = & \fh{ (U_1 \otimes \identity_\sB) \, \ket{\phi_1(z,j)} } 
                  { \ket{\phi_1(z^{(l)}, j)} }
    \quad = \quad h_1(j,l,z) \enspace.
\end{eqnarray*}
In other words, the state on the input of type~``01'' is, up to a
unitary operation on Alice's part, ``close'' to that on the input of
type~``11''.

\begin{figure}[ht]
\centerline{\includegraphics[width=450pt]{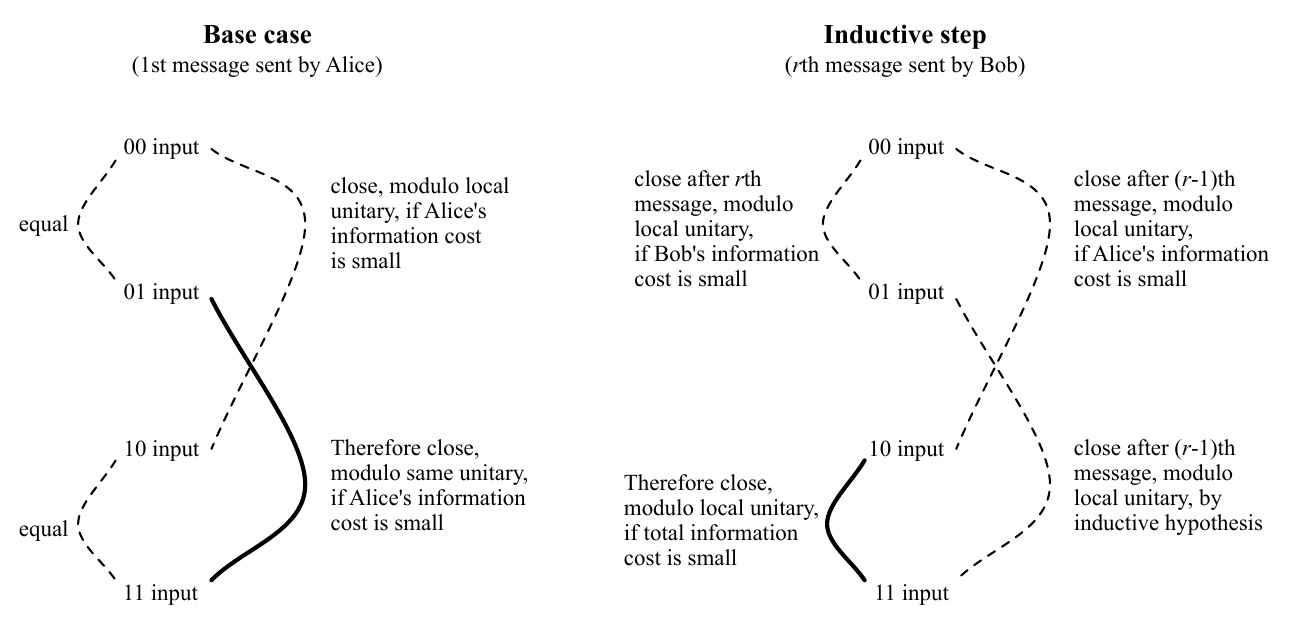}}
\caption{ \label{fig-proof}
The relationship between states at intermediate 
stages of the protocol, as described in the proof of Lemma~\ref{lem:key}.
}
\end{figure}

We prove that the lemma holds for~$r$, assuming that it holds 
for~$r - 1 \in [t]$. The argument here follows the same intuition as in
the base case, but is more involved because the analogous equalities 
need not hold.
However, the first pair of states may be shown to be close to each 
other, modulo a local unitary operator, by virtue of Bob's low 
information cost. The second pair are assumed to be close, again 
modulo a local unitary operator, by the inductive hypothesis. A careful 
hybrid argument then gives us the claimed bound. Figure~\ref{fig-proof}
depicts this schematically.

There are two cases: $r$ is odd, or~$r$ is even. We conduct the argument 
in the second case, when~$r$ is even.  The argument for~$r$ odd is 
similar, and is omitted.

By our convention, Bob sends the even numbered messages, including
the~$r$th message.
By Lemma~\ref{thm-pairs-quantum}, the states on the inputs of
type~``$00$'' and~``$01$'' are ``close'' up to the local
unitary~$U_r$, i.e., 
\begin{eqnarray}
\label{eqn-pairs-r}
\fh { \, ( \identity_\sA \otimes U_r) \, \ket{\phi_r(z, j)} }
        { \ket{\phi_r(z, l)} }
    & = & h_r(j,l,z) \enspace. 
\end{eqnarray}
Similarly, by Lemma~\ref{thm-flip-quantum}, the states \emph{before\/}
the~$r$th message on the inputs of type~``$00$'' and~``$10$'' are 
``close'' up to the local unitary~$U_{r-1}$, i.e.,
\begin{eqnarray}
\label{eqn-00-10}
\fh{ (U_{r-1} \otimes \identity_\sB) \, \ket{\phi_{r-1}(z,j)} }
             { \ket{\phi_{r-1}(z^{(l)}, j)} }
    & = & h_{r-1}(j,l,z) \enspace.
\end{eqnarray}
By the induction hypothesis, we also have the following relationship
between the states on inputs of type~``$01$'' and~``$11$'':
\begin{eqnarray}
\label{eqn-01-11}
\fh{ (U_{r-1} \otimes \identity_\sB) \, \ket{\phi_{r-1}(z,l)} }
             { \ket{\phi_{r-1}(z^{(l)}, l)} }
    & \leq & h_{r-1}(j,l,z) + 2 \sum_{i = 1}^{r-2}
                   h_i(j,l,z) \enspace .
\end{eqnarray}
Now
\begin{eqnarray*}
\ket{\phi_r(z,l)}
    & = & (\identity_\sA \tensor V_{r, z[1,l]} ) \,
                  \ket{\phi_{r-1}(z,l)} 
                  \enspace, \qquad \textrm{and} \\
\ket{\phi_r(z^{(l)},l)}
    & = & (\identity_\sA \tensor V_{r, z[1,l]} )  \,
                  \ket{\phi_{r-1}(z^{(l)},l)} \enspace,
\end{eqnarray*}
where~$V_{r, z[1,l]}$ is the unitary operator that Bob applies on his
part of the state (i.e., on the register holding state~$Q_{r-1}$
before sending the~$r$th message. Note that~$V_{r,z[1,l]}$
commutes with~$U_{r-1}$, as they act on disjoint sets
of qubits. Since the Bures distance is invariant under unitary operators,
Eq.~(\ref{eqn-00-10}) gives us
\begin{eqnarray}
\label{eqn-flip-r1}
\fh{ (U_{r-1} \otimes \identity_\sB) \, \ket{\phi_r}(z,j)}
             { \ket{\phi_r(z^{(l)}, j)} }
    & = & h_{r-1}(j,l,z) \enspace, 
\end{eqnarray}
and Eq.~(\ref{eqn-01-11}) gives us
\begin{eqnarray}
\label{eqn-flip-r2}
\fh{ (U_{r-1} \otimes \identity_\sB) \, \ket{\phi_r(z,l)} }
             { \ket{\phi_r(z^{(l)}, l)} }
    & \leq & h_{r-1}(j,l,z) + 2 \sum_{i = 1}^{r-2} h_i(j,l,z) \enspace .
\end{eqnarray}
By the Triangle Inequality, Eqs.~(\ref{eqn-pairs-r}), (\ref{eqn-flip-r1}),
and~(\ref{eqn-flip-r2}), and the observation that~$U_{r-1}$
and~$U_r$ act on disjoint sets of qubits, we get
\begin{eqnarray*}
\lefteqn{
\fh{ ( \identity_\sA \otimes U_r) \, \ket{\phi_r(z^{(l)}, j)} }
        { \ket{\phi_r(z^{(l)}, l)} }
} \\
    & \le & \fh{(\identity_\sA \otimes U_r) \, \ket{\phi_r(z^{(l)}, j)}}
           {(U_{r-1} \tensor \identity \otimes U_r) \, \ket{\phi_r(z, j)}}
            \\
    &     & \mbox{} + \fh{(U_{r-1} \tensor \identity \otimes U_r) \, 
            \ket{\phi_r(z, j)}}{ \ket{\phi_r(z^{(l)}, l)} } \\
    & =   & h_{r-1}(j,l,z) + \fh{(U_{r-1} \tensor \identity \otimes U_r) \, 
            \ket{\phi_r(z, j)}}{ \ket{\phi_r(z^{(l)}, l)} } \\
    & \le &  h_{r-1}(j,l,z)
             + \fh{(U_{r-1} \tensor \identity \otimes U_r) \, 
            \ket{\phi_r(z, j)}}{ (U_{r-1} \tensor \identity_\sB)  \, 
            \ket{\phi_r(z, l)}} \\
    &     & \mbox{} + \fh{(U_{r-1} \tensor \identity_\sB) \, 
            \ket{\phi_r(z, l)}}{ \ket{\phi_r(z^{(l)}, l)} } \\
    & \le & h_{r-1}(j,l,z) + h_r(j,l,z) 
            + \fh{(U_{r-1} \tensor \identity_\sB) \, 
            \ket{\phi_r(z, l)}}{ \ket{\phi_r(z^{(l)}, l)} } \\
    & \le & h_{r}(j,l,z) + 2 \sum_{i = 1}^{r-1} h_i(j,l,z) \enspace .
\end{eqnarray*}
(The identity operators without a subscript in this derivation act on
the space of the~$r$th message.)  This completes the induction step.
\end{proof}

\section{Concluding remarks}
\label{sec-final}

The main focus of this article is the amount of information two parties
necessarily reveal about their inputs in the process of the computing a
function in a distributed manner. The function of interest is \AIndex,
a natural variant of the \Index\ function that is ubiquitous in
communication complexity. We show that in
any randomized communication protocol that computes this function 
correctly with constant error on the uniform distribution (a ``hard'' 
distribution), either Alice reveals~$\Omega(n)$ information about 
her $n$-bit input, or Bob reveals~$\Omega(1)$ information about his
$(\log n)$-bit input, even when the inputs are drawn from the uniform 
distribution over inputs which evaluate to~$0$. 
At first glance, a trade-off under a distribution on inputs on which 
the function value is \emph{known in advance\/} may appear to be 
counter-intuitive. This is a consequence of the correctness of the
protocol on the hard distribution. Such a phenomenon was first
demonstrated by Bar-Yossef, Jayram, Kumar, and 
Sivakumar~\cite{Bar-YossefJKS04}.

The motivation for this work comes from the study of tasks that may 
be accomplished with a few sequential scans of massive data,
using significantly smaller memory, i.e., through streaming algorithms.
The above result has implications for the space required by streaming 
algorithms for \dyck(2), the problem of checking the syntax of a 
parenthesized expression. It implies that for this problem, we need 
space~$\sqrt{n}/T$ on inputs of length~$n$, when allowed~$T$ unidirectional 
passes over the input.

The proof of the information cost trade-off showcases a modular and
conceptually simple technique involving the Average Encoding Theorem and the
Cut-and-Paste Lemma. Originally developed to analyse properties of
quantum protocols, Average Encoding has been used more widely in
classical complexity theory. For instance, it has been used to derive
lower bounds for data structures~\cite{SenV08}, and can be 
used to derive the ``Disguising Distribution Lemma''~\cite{Drucker12},
which has applications for instance compression. The technique developed
in this article has also been adapted by Fran{\c c}ois and
Magniez to prove space lower bounds for the problem of checking priority
queues with time stamps in the streaming model~\cite{FrancoisM13}. We
expect that these tools have yet more applications in information
processing.

A few recent works show how simple \emph{quantum\/} streaming
algorithms may use exponentially smaller amount of space as compared
with classical ones~\cite{LeGall09,GavinskyKKRW08}. We ask if there is
similar advantage in solving a natural and important problem such as
\dyck(2). We make partial progress in this direction, by establishing 
a \emph{quantum\/} information cost trade-off for \AIndex. We show that
in quantum protocols that compute \AIndex\ correctly with constant 
error on the uniform distribution, either Alice reveals~$\Omega(n/t)$
information, or Bob reveals~$\Omega(1/t)$ information, where~$t$ is 
the number of messages in the protocol, even when the inputs are 
drawn from the aforementioned easy distribution.

The quantum information cost trade-off by itself does not imply a space
lower bound for streaming quantum algorithms. The reduction from
streaming algorithms for $\dyck(2)$ with small space to quantum 
two-party protocols for \AIndex\ breaks down for the notion of
information cost we adopt. We conjecture a trade-off similar to
Theorem~\ref{thm-main-quantum} for the notion of information cost
in Eqs.~(\ref{eqn-qic1}) and~(\ref{eqn-qic2}). We leave the resolution 
of this conjecture as an intriguing open problem.


\end{document}